\documentclass[11pt]{article}

\usepackage{geometry}
\usepackage{enumitem}

\usepackage{amssymb}
\usepackage{mathrsfs}
\usepackage{amsmath}
\usepackage[latin1]{inputenc}

\newtheorem{theorem}{Theorem}[section]
\newtheorem{open}{Open Question}
\newtheorem{corollary}[theorem]{Corollary}
\newtheorem{lemma}[theorem]{Lemma}
\newtheorem{proposition}[theorem]{Proposition}
\newtheorem{claim}[theorem]{Claim}
\newtheorem{definition}[theorem]{Definition}
\newtheorem{remark}[theorem]{Remark}

\def\squarebox#1{\hbox to #1{\hfill\vbox to #1{\vfill}}}
\newcommand{\qed}{\hspace*{\fill}
\vbox{\hrule\hbox{\vrule\squarebox{.667em}\vrule}\hrule}\smallskip}

\newenvironment{proof}{\noindent{\bf Proof:~~}}{\(\qed\)}

\begin{document}

\title{Computational Efficiency Requires Simple Taxation}

\author{Shahar Dobzinski\thanks{Weizmann Institute of Science.}}

\maketitle

\begin{abstract}
We characterize the communication complexity of truthful mechanisms. Our departure point is the well known taxation principle. The taxation principle asserts that every truthful mechanism can be interpreted as follows: every player is presented with a menu that consists of a price for each bundle (the prices depend only on the valuations of the other players). Each player is allocated a bundle that maximizes his profit according to this menu. We define the \emph{taxation complexity} of a truthful mechanism to be the logarithm of the maximum number of menus that may be presented to a player.

Our main finding is that in general the taxation complexity essentially equals the communication complexity. The proof consists of two main steps. First, we prove that for rich enough domains the taxation complexity is at most the communication complexity. We then show that the taxation complexity is much smaller than the communication complexity only in ``pathological'' cases and provide a formal description of these extreme cases.

Next, we study mechanisms that access the valuations via value queries only. In this setting we establish that the menu complexity -- a notion that was already studied in several different contexts -- characterizes the number of value queries that the mechanism makes in exactly the same way that the taxation complexity characterizes the communication complexity. 

Our approach yields several applications, including strengthening the solution concept with low communication overhead, fast computation of prices, and hardness of approximation by computationally efficient truthful mechanisms. 
\end{abstract}

\thispagestyle{empty}
\newpage \setcounter{page}{1}

\section{Introduction}

The field of Communication Complexity studies settings in which $n$ players are interested in computing some known function $f$. Each player $i$, holds some input $x_i$. The basic task is to determine the maximum number of bits that the parties need to exchange in order to compute $f(x_1,\ldots, x_n)$. 

One of the most successful applications of communication complexity is Algorithmic Mechanism Design, starting with the pioneering work of Nisan and Segal \cite{N02, NS06}. Nisan and Segal proved lower bounds on the approximation ratios achievable by algorithms with low communication complexity for combinatorial auctions. Many other applications of communication complexity to mechanism design have been introduced since. For example, communication complexity is used to bound the power of certain computationally-efficient truthful mechanisms \cite{DN07a, DSS15}, to understand the overhead of price computation \cite{FS09, BBNS08}, and to prove bounds on the quality of equilibria \cite{R14}. 

Our goal in this paper is to answer a fundamental question in the intersection of communication complexity and algorithmic mechanism design: given a truthful mechanism $A$, how many bits do the parties need to exchange in order to determine the allocation and payments?

A bit more formally, the communication complexity of a protocol is the maximum number of bits that are exchanged in the protocol, where the maximum is taken over all inputs. The communication complexity of a function  $f$, denoted $cc(f)$, is the communication complexity of the protocol that computes $f$ with the smallest communication complexity.

A truthful mechanism $A$ is composed of a social choice function that selects one alternative from a set $\mathcal S$ of alternatives and a payment function that specifies the payment of each player. Our results build on a basic concept of Mechanism Design, the \emph{taxation principle}. Denote by $v_i(S)$ the value of player $i$ for alternative $S\in \mathcal S$. The taxation principle asserts that $A$ can be interpreted in the following simple form: every player $i$ is (implicitly) presented with a menu $\mathcal M_{v_{-i}}$ that is a function that assigns a price (possibly $\infty$) for each alternative in $\mathcal S$. $\mathcal M_{v_{-i}}$ depends only on the valuations $v_{-i}$ of the other players. The truthful mechanism $A$ always outputs an alternative $S$ that simultaneously maximizes the profit $v_i(S)-\mathcal M_{v_{-i}}(S)$ of each player $i$.

With this interpretation in mind, given a truthful mechanism $A$, for each player $i$ denote by $M^i=\{\mathcal M_{v_{-i}}\}_{v_{-i}}$ the set of menus that might be presented to $i$. Denote by $tax(A)$ the \emph{taxation complexity} of a truthful mechanism $A$ -- the number of bits needed to represent an index of a specific menu among the set of menus that may be presented to a player. That is, $tax(A)=\max_i\log |M^i|$.



Our main finding directly connects the semantics of the mechanism and its communication complexity by showing that the taxation complexity of every truthful mechanism essentially equals its communication complexity:

\vspace{0.07in} \noindent \textbf{Informal take-home message of this paper: } In ``rich enough'' domains, $tax(A)\approx cc(A)$. 

\vspace{0.07in} \noindent We also apply the lens of the taxation principle in a more restricted model and prove an analogous result: if access to the valuations is restricted to value queries, then the menu complexity essentially equals the query complexity (see definitions below).

In the rest of the introduction we provide a more formal description of the setup\footnote{See Section \ref{sec-preliminaries} for formal definitions.}, of the results, and of various implications.

\subsection*{The Setting}

For concreteness, this paper considers only the setting of combinatorial auctions, although it should be possible to extend the results beyond that domain. In a combinatorial auction there is a set $M$ ($|M|=m$) of heterogeneous items and a set $N$ ($|N|=n$) of players. The output is an allocation $(S_1,\ldots, S_n)$ of the items to the players. The private information of each player $i$ is his value for the set of items he receives: $v_i:2^M\rightarrow \mathbb{R}$. As common in the literature, in this paper we assume that each $v_i$ is normalized ($v_i(\emptyset)=0$) and monotone (for each $S\subseteq T$, $v_i(T)\geq v_i(S)$).

This paper considers truthful mechanisms. The Algorithmic Mechanism Design literature usually defines truthful mechanisms to be those that implement some social choice function in a dominant-strategy equilibrium. In this paper a mechanism is \emph{truthful} if it implements the social choice function in an ex-post Nash equilibrium. This solution concept applies to games with \emph{incomplete} information and is closely related but less restrictive than dominant-strategy equilibrium (this makes our results only stronger). It is more appropriate for the iterative mechanisms that this paper considers. Very roughly speaking, in an ex-post Nash equilibrium a dominant strategy of every player is to play according to his true valuation, as long as the other players are not playing ``crazy'' strategies. We refer the reader to Section \ref{sec-preliminaries} for formal definitions and discussion. 

\subsection*{The Taxation Principle in Algorithmic Mechanism Design}

The taxation principle \cite{Hammond79, G81} was already considered in Algorithmic Mechanism Design. Most notably, the crux of the impossibility results of \cite{D11, DughmiV11,DV12} is showing that for every mechanism that approximately maximizes the welfare there must be an instance in which one player is presented a ``complicated'' menu. In particular, finding a profit-maximizing bundle in that menu is hard. 

The paper \cite{HN13} considers a setting with only a single player whose valuation is drawn from some known distribution. Since there is only one player, the taxation principle implies that all a truthful mechanism can do is to present a fixed menu to the player. The player then ``selects'' a profit-maximizing bundle. They show that to approximately maximize revenue the prices of many alternatives (equivalently, bundles) in that menu must be finite (high ``menu complexity'').

We stress that the notion of taxation complexity does not measure the difficulty of finding a profit-maximizing bundle, nor how difficult it is to represent a specific menu. The taxation complexity takes a more ``high-level'' view of the mechanism and only measures the \emph{number} of menus that might be presented to a player.

\subsection*{The Taxation Complexity is at most the Communication Complexity}

Our first main result says that in rich enough domains the taxation complexity is at most the communication complexity:
$$
tax(A)\leq cc(A)
$$
We will shortly provide a more formal statement, but to better understand this result and its implications, we first discuss a domain for which this result \emph{does not} hold. Consider a two-player combinatorial auction where the valuations of the players belong to the class of gross substitutes (GS) valuations\footnote{The definition of this class is subtle; Since it will not be needed in this paper, we refer the interested reader to the survey \cite{PL14} for a definition.}. Let us furthermore restrict the values of all bundles to be integers in $\{1,\ldots, m\}$. We study the VCG mechanism in that setting. Since the welfare maximizing allocation can be found with $poly(m)$ communication \cite{NS06} (this is also implied by the algorithms mentioned in \cite{PL14}), the communication complexity of the VCG mechanism is $poly(m)$ as well.

Let us now analyze the taxation complexity of the mechanism. Denote the valuation of player $1$ by $v$. Player $1$ presents a menu to player $2$. By the definition of the VCG mechanism, the price of bundle $S$ in that menu is $v(M)-v(M-S)$. Thus, there is a one-to-one and onto correspondence between the set of possible valuations of player $1$ and the set of menus he presents to player $2$. All that is left is to point out that the number of gross substitutes valuations is\footnote{In fact, Knuth \cite{K74} shows that the number of matroid rank functions on $\{1,\ldots, m\}$ is doubly exponential, and it is known that every matroid rank function is in particular gross substitutes.} doubly exponential (e.g, Knuth \cite{K74} shows that it is at least $ {2^{\frac {{m \choose {m/2}}} {2m}}} /{m!}$). Therefore, for two players with gross substitutes valuations, the taxation complexity of the VCG mechanism $tax(VCG)$ is therefore exponential, whereas the communication complexity $cc(VCG)$ is polynomial.

In contrast, in richer domains the taxation complexity is not much larger than the communication complexity:

\vspace{0.07in} \noindent  \textbf{Theorem: } Fix some mechanism $A$. 
\begin{enumerate}[noitemsep]
\item If $A$ is truthful for general valuations, then $tax(A)\leq cc(A)$.
\item If $A$ is truthful for subadditive valuations, then also $tax(A)\leq cc(A)$.
\item If $A$ is truthful for XOS valuations, then $tax(A)\leq m\cdot (cc(A)+1)$.
\item If $A$ is truthful for submodular valuations, then $tax(A)\leq d\cdot m\cdot (cc(A)+1)$, where $d=|\{\mathcal M(S)\}_{\mathcal M, S}|$ is the total number of distinct prices that appear in some menu.
\end{enumerate}
\noindent The cautious reader might wonder how it can be that the class of GS valuations is contained in all the above-mentioned classes, but in these classes the communication complexity severely limits the taxation complexity. The point is that implementations of mechanisms that are specifically tailored to GS valuations are able to find profit-maximizing bundles without learning the full menu. However, when running those implementations in richer domains, the set of possible deviations increases and truthfulness is lost. 

\subsection*{Characterizing the Communication Complexity of Truthful Mechanisms}

We have that in rich enough domains $tax(A)\leq poly(cc(A))$. Had we were able to prove that $cc(A)\leq poly(tax(A))$, this would immediately imply that the communication complexity of a truthful mechanism is completely determined (up to polynomial factor) by $tax(A)$ -- a well defined combinatorial property that depends only on the social choice function. 

However, we show that $cc(A)$ cannot be bounded by $poly(tax(A))$. Towards this end, consider the following (naive and incorrect) implementation of a two-player mechanism $A$ which is truthful for general valuations: since the menu that is presented to a player depends only on the other player's valuation, each player can send $tax(A)$ bits that denote the index of the menu he presents to the other player. The obvious next step is to ask each player to select the profit maximizing bundle from the menu that was presented to him, announce it (using additional $m$ bits) and allocate accordingly. The total communication cost of this implementation is $2(tax(A)+m)$ as we wanted.

The above implementation is incorrect since this last step is not well defined because of tie-breaking: there might be several bundles that simultaneously maximize the profit. The tie-breaking rule that defines which profit-maximizing bundle each player receives can be in principle quite involved, and there is no way to avoid that: we show (Subsection \ref{subsec-low-tax-high-comm}) a two-player mechanism with taxation complexity $1$ and communication complexity $exp(m)$, due to tie breaking. 

This leads us to the following definition. Let $tie(A)$ be the communication complexity of determining the allocation in a truthful mechanism $A$, where the input of a player is a valuation $v_i$ and in addition all players know the menu that is presented to each player. Notice that obviously $tie(A)\leq cc(A)$ since we can always ignore this extra information and simply run $A$ to determine the allocation. By our discussion above and our bound on the taxation complexity we get that:

\vspace{0.07in}\noindent \textbf{Theorem: } Let $A$ be a two-player truthful mechanism for rich enough domain. Then,
$$\frac {tax(A)+tie(A)} 2\leq cc(A)\leq 2(tax(A)+m)+tie(A)$$
\noindent In particular, whenever the communication complexity of the tie-breaking rule is low, we indeed get that the communication complexity almost equals the taxation complexity.

Can we extend this theorem to more than two players? The missing component for three players or more is that it is not clear whether it is possible to \emph{explicitly} find the menu that $n-1$ players present to the remaining player with low communication (the taxation principle only guarantees the \emph{existence} of such menu, but gives no guidance on how to find it). We provide a positive answer:

\vspace{0.07in}\noindent \textbf{The Menu Reconstruction Theorem: } Let $A$ be an $n$-player truthful mechanism. Fix some player $i$. Then for every valuation profile $v_{-i}$, there is a protocol with communication complexity $poly(tax(A),price(A), m,n)$ that finds the menu that is presented to player $i$.

\vspace{0.07in}\noindent Where we denote by $price(A)$ the maximum number of bits that it takes for any $n-1$ players to find the price of a given bundle $S$ in the menu that they present to the remaining player. Therefore:

\vspace{0.07in}\noindent \textbf{Theorem: } Let $A$ be an $n$-player truthful mechanism in any domain. Then,
$$ cc(A)\leq poly(tax(A),price(A),tie(A), m,n)$$
\noindent We also show that as long as the domain includes additive valuations, $price(A)\leq cc(A)$. We thus get that for rich enough domains $\frac {tax(A)+tie(A)+price(A)} 3\leq cc(A)$, which allows us to completely determine the communication complexity of truthful mechanisms (up to polynomial factors):
$$\frac {tax(A)+price(A)+tie(A)} {3}\leq cc(A) \leq poly(tax(A),price(A), tie(A),m,n)$$
\noindent Our characterization is tight in the sense that if we drop at least one of the three main terms ($tax(A),price(A), tie(A)$) then the gap between the LHS and the RHS might be exponential. For instance, we have already mentioned an example of a truthful mechanism $A$ with $tax(A)=1$ (and thus $price(A)=0$) in which $cc(A)=exp(m)$. We also provide other ``pathological'' examples with similar gaps when dropping either $price(A)$ or $tax(A)$.

\subsection*{Characterizing the Query Complexity}

Up until now we imposed no restrictions on the communication between the parties. However, many of the truthful mechanisms in the literature assume that the valuations are represented as black boxes that answer only a specific type of queries. A simple type of query that was extensively studied is a \emph{value query}: given a bundle $S$, what is $v(S)$? We now characterize the number of value queries that a truthful mechanism makes, again by applying the taxation principle.

Denote the query complexity of a truthful mechanism by $val(A)$ -- this is the number of value queries that the most efficient implementation of $A$ makes. Following \cite{HN13}, The \emph{menu complexity} of $A$, denoted $mc(A)$, is roughly speaking the maximum number of bundles with finite price that appear in any menu that is used in\footnote{The description of the menu complexity given here is inaccurate as it ignores tie breaking issues. We refer the reader to the technical parts of this paper for the precise definition.} $A$. We establish that for mechanisms that use only value queries, the menu complexity characterizes the query complexity in exactly the same way that the taxation complexity characterizes the communication complexity:

\vspace{0.07in}\noindent \textbf{Theorem: } Let $A$ be a mechanism that is truthful for general valuations and accesses the valuations via value queries only. Then:
$$
\frac {mc(A)+price^{val}(A)+tie^{val}(A)} 3\leq val(A)\leq poly(mc(A),price^{val}(A),tie^{val}(A), m,n)
$$
\noindent where $price^{val}(A)$ and $tie^{val}(A)$ are defined similarly to $price(A)$ and $tie(A)$ with the additional restriction that the communication is restricted to value queries. We note that the inequality $mc(A)\leq tax(A)$ is in fact implicit in \cite{D11} and \cite[Theorem 11]{N14}, whereas the right inequality (a menu reconstruction theorem that uses only value queries) is new and very different from the menu reconstruction theorem for unrestricted communication.

To strengthen the analogy between $tax(A)$ and $mc(A)$, consider the following two player menu optimization problem: Alice's input is some menu $\mathcal M\in U$, where the set of menus $U$ is known in advance. Bob's input is some valuation $v$. The goal is to find a bundle that maximizes the profit $v(S)-\mathcal M(S)$. We restrict ourselves to one way protocols: Alice speaks first and then Bob. After Bob speaks, both parties know a profit maximizing bundle $S$.

We observe that if we let $U$ be the set of menus presented to some player in a truthful mechanism $A$, the (one way) communication complexity of the menu optimization problem is $tax(A)$ (up to an additive factor of $m$ bits). Interestingly, our results yield that when Bob is restricted to value queries (Alice sends an arbitrary message, then, based on this message, the center queries for the value of some bundles in $v$) then the communication complexity is essentially $mc(A)$. That is, both $tax(A)$ and $mc(A)$ capture the informational bottleneck of finding a profit-maximizing bundle in a non-interactive way in their respective models.

The other type of popular query is \emph{demand query}: given prices $p_1,\ldots, p_m$ return a bundle $S\in \arg\max_Tv(T)-\Sigma_{j\in T}p_j$. We identify the \emph{affinity} of a mechanism as the communication complexity of the menu optimization problem when the communication is restricted to demand queries. Specifically, a menu $\mathcal M$ is \emph{$\alpha$-min affine} if there are $\alpha$ price vectors $p^1,\ldots, p^\alpha$ and $\alpha$ non-negative numbers $r^1,\ldots, r^\alpha$ such that for all $S$, $\mathcal M(S)=\min_{1\leq k\leq \alpha}\Sigma_{j\in S}(p^k_j)+r^k$. The affinity of $A$, denoted $aff(A)$, is the maximal number $\alpha$ such that all menus presented by $A$ are $\alpha$-min affine. 

Denote by $dem(A)$ the number of demand queries that the most efficient implementation of a truthful mechanism $A$ makes. Just as $tax(A)\leq cc(A)$ and $mc(A)\leq val(A)$, we show that for general valuations $aff(A)\leq dem(A)$. However, unlike $tax(A)$ and $mc(A)$, $aff(A)$ cannot be used to characterize $dem(A)$. This is one particular consequence of an \emph{impossibility} result:

\vspace{0.07in}\noindent \textbf{Theorem (no menu reconstruction theorem for demand queries):} There is a mechanism $A$ that is truthful for some player $i$ with a general valuation in which $dem(A)=poly(m)$ and $price^{dem}(A)=1$, but $exp(m)$ demand queries are needed to find the menu presented to player $i$.

\vspace{0.07in}\noindent In other words, just as $q$ value queries suffice to find a profit maximizing bundle when the menu complexity is $q$, when the affinity is $q$ a profit maximizing bundle can be found with $q$ demand queries. However, by our impossibility result it is not easy to figure out \emph{which} $q$ queries to make. Nevertheless, the affinity $aff(A)$ will be useful in some of the applications that we mention below.

The query complexity of mechanisms is studied in Section \ref{sec-query-complexity}.

\subsection*{Implications and Extensions}

We view the study of these complexity measures as an investigation of the fundamentals of Algorithmic Mechanism Design that needs no further justification. Nevertheless, as is often the case, studying the foundations yields several interesting implications. We elaborate on some now, as well as on some open questions. Other open questions are stated in the technical parts of the paper. 

\paragraph{From ex-post Nash to Dominant Strategies (Section \ref{subsec-to-dominant}).} The revelation principle implies that  any mechanism that implements a social choice function in an ex-post Nash equilibrium can be transformed to a mechanism that implements the same social choice function in a dominant-strategy equilibrium. However, the communication blow-up might be exponential, as in combinatorial auctions with GS valuations. Our theorems allow us to obtain a new transformation that takes any two-player mechanism that implements an ex-post Nash equilibrium in a rich enough domain to a mechanism that implements the same social choice function in dominant-strategy equilibrium with only a polynomial blow up in the communication complexity.

\paragraph{The Limits of Computationally-Efficient Truthful Mechanisms (Section \ref{subsec-welfare}).}

A major research direction in Algorithmic Mechanism Design studies the power of computationally efficient truthful mechanisms for welfare maximization in combinatorial auctions (e.g., \cite{LOS:J,LMN03,MN02,LS05,BKV05,DNS05}). For example, VCG is a truthful mechanism that maximizes the welfare, but in combinatorial auctions with submodular valuations it requires exponential communication \cite{NS06}. On the other hand, there is a $1.58$-approximation algorithm that uses only polynomial communication \cite{FV06}, but it is not truthful. The best deterministic truthful mechanism that uses only polynomial communication achieves a poor approximation ratio of $O(\sqrt m)$ \cite{DNS05}. Whether this is the best possible for a deterministic truthful mechanism that uses only polynomial communication is a major open question. This state of affairs is typical to other problems as well, and we basically completely lack tools for proving impossibility results for computationally efficient truthful mechanisms\footnote{The papers \cite{D11,DughmiV11, DV12} prove impossibility results when access is restricted to value queries, or when the valuations are given in a succinct and very specific form. For the general communication model, or even when access is restricted to demand queries, that are no impossibility results on the power of computationally efficient truthful mechanisms.}.  Our work gives rise to two different novel approaches for proving such impossibilities.

\vspace{0.07in} \noindent \emph{Approach I: simultaneous algorithms.} The first approach for proving impossibilities is by a reduction to \emph{simultaneous} algorithms. This model was introduced in \cite{DNO14}: each of the $n$ players simultaneously sends $c$ bits that are a function of his valuation only. The center then determines the allocation using only those messages. We show that \emph{impossibilities for two-player simultaneous algorithms imply impossibilities for computationally efficient truthful mechanisms}. Specifically, we show that if there is a truthful mechanism (for the domains discussed above) that provides an approximation ratio of $\alpha$ with communication complexity $cc(A)$, then there is an $\alpha$-approximation \emph{simultaneous} algorithm where the length of the simultaneous messages is $poly(cc(A))$. For example, a proof that no simultaneous algorithm with polynomially long messages for submodular players achieves a $1.59$ approximation (ignoring incentives issues) immediately establishes the first ever gap between computationally efficient truthful mechanisms and their non-truthful counterparts. Note that strong impossibility results for simultaneous algorithms are known \cite{DNO14}, but those unfortunately hold only for a large number of players. In particular, nothing is known for two players.

\vspace{0.07in} \noindent \emph{Approach II: lower bounds on $tax(A)$.} The second approach involves handling the taxation complexity directly. The idea is simple: suppose one can prove that the taxation complexity of every truthful $\alpha$-approximation mechanism for combinatorial auctions with submodular players is exponential (ignoring computational issues). Since $tax(a)\leq poly(cc(A))$, it follows that the communication complexity of every $\alpha$-approximation truthful mechanism is exponential as well, which establishes a gap between the power of truthful and non-truthful computationally efficient algorithms.

This approach can be extended to mechanisms with restricted access. For example, to prove impossibility results for mechanisms that use only demand queries (e.g., \cite{D16}) it suffices to show that the affinity of every $\alpha$-approximation truthful mechanism for combinatorial auctions with submodular players (ignoring computational issues) is exponential. We note that the menu complexity of $m^{\frac 1 2-\epsilon}$-approximation mechanisms for submodular valuations was already proved to be exponential, which indeed yielded an impossibility on the power of truthful mechanisms that use value queries in the aforementioned setting (the direct hardness approach of \cite{D11}).

\vspace{0.07in} \noindent \emph{Extensions to randomized mechanisms.} Before this paper, the only viable approach for proving impossibility results for randomized mechanisms in the general communication model was by characterizing all truthful mechanisms with a good approximation. Such characterizations are notoriously hard even for deterministic mechanisms. For randomized ones, it is probably fair to describe the possibility of obtaining such characterizations in the foreseeable future as almost hopeless.

There are two main notions of randomized truthfulness. The first is truthfulness in expectation, where each player maximizes his \emph{expected} profit. We discuss this notion below, and here we focus on the other (stronger) notion: universal truthfulness. Universally truthful mechanisms are simply a probability distribution over deterministic mechanisms. Interestingly, universally truthful mechanisms achieve the best currently known approximation ratios in many important settings (e.g., combinatorial auctions with submodular players \cite{D16} and with subadditive players \cite{D07}), even if truthful in expectation mechanisms are considered. 

Both approaches are capable of proving impossibility results for universally truthful mechanisms. First, an impossibility for two-player randomized simultaneous algorithms implies an impossibility for randomized truthful mechanisms. The second approach is also applicable: a lower bound on the taxation complexity is likely to be proved by obtaining a distribution over the input on which no mechanism with polynomial taxation complexity provides a good approximation ratio. Yao's principle and our results imply that for every universally truthful mechanism with polynomial communication there is an instance on which its (expected) approximation ratio is bad.

\paragraph{Efficient Price Computation.}

Every truthful mechanism has two tasks: the first is to compute the social choice function and the second is to compute the players' payments. Fadel and Segal \cite{FS09} ask whether the additional communication cost of computing the payments is significantly larger than the communication complexity of computing the social choice function. In deterministic settings, they show that the bound is at most exponential and ask whether this is tight. Single parameter domains are handled by \cite{BBNS08} and various multi-parameter domains are handled by \cite{WS15, BKS13} via a ``single call' approach, but at the cost of introducing randomization. 

We extend this line of work. Since we showed that in rich enough domains $tax(A),price(A)\leq cc(A)$, our menu reconstruction theorem immediately implies that fully presenting the players with the actual menu takes only $poly(cc(A))$ bits, thus finding the menu is essentially as easy as computing the allocation and the payments of winning bundles ($cc(A)$ bits).

\paragraph{Truthful in Expectation Mechanisms (Section \ref{subsec-truthful-in-expectation}).}

Our theorems were proved for deterministic mechanisms (and thus they also apply to randomized universally truthful mechanisms). Since truthful in expectation mechanisms were extensively studied (e.g., \cite{LS05,DD09,DRY11}), it is natural to ask whether in truthful in expectation mechanisms $tax(A)$ similarly characterizes $cc(A)$. We provide a negative answer: for combinatorial auctions with general valuations there is a truthful in expectation mechanism with polynomial communication and exponential taxation complexity.

Recall that another setting where the taxation complexity might be exponential in the communication complexity is when the valuations are gross substitutes. One can further make the following wild speculation, which lacks more evidence and formalization: in domains where the gap between the taxation complexity and the communication complexity is small, the performance of computationally efficient truthful mechanisms is poor, whereas for domains where it is large, truthfulness is not a severely limiting requirement.

%

\section{Bounding the Taxation Complexity: $tax(A)\leq poly(cc(A))$}\label{sec-bound-tax-complexity}

We would now like to bound the taxation complexity as a function of the communication complexity. As a warm up, in Subsection \ref{subsec-warmup} we prove that the taxation complexity of a truthful mechanism for general valuations with communication complexity $cc(A)$ is $cc(A)+1$, as long as all bundles get a finite price in every menu (this is true if, for example, the mechanism is deterministic and provides some finite approximation ratio to the welfare).

In Subsection \ref{subsec-main} we strengthen this result in several aspects. First, we generalize the result to any truthful mechanism by allowing the prices of bundles to be $\infty$. Second, we slightly strengthen the bound on the taxation complexity to $cc(A)$ (and not just $cc(A)+1$). Finally, we extend our results for other classes of valuations, not just general valuations.

\subsection{A Warm-up}\label{subsec-warmup}

\begin{theorem}[warm-up theorem, inferior to the result of Subsection \ref{subsec-main}]\label{thm-warm-up}
Consider a truthful mechanism $A$. Suppose that for each player $i$, bundle $S$, and menu $\mathcal M$ that might be presented to player $i$ we have that $\mathcal M(S)<\infty$. Then, $tax(A)\leq cc(A)+1$.
\end{theorem}
\begin{proof}
Fix some player $i$. Let the \emph{taxation complexity of player $i$} be the logarithm of the number of menus that player $i$ might be presented with in $A$. We will prove that the taxation complexity of player $i$ is at most $cc(A)+1$ and the theorem will follow. 

Define a protocol $A$' that is completely identical to the most efficient implementation of $A$ (in particular, the same allocation and payment functions), except that at the end of $A'$ player $i$ sends the bit $1$ if his profit is positive (i.e., $A_i(v)-p_i(v)>0$) and the bit $0$ otherwise. Note that $A'$ is truthful since $A$ is truthful and that that the communication complexity of $A'$ is ${cc(A)+1}$. In addition, for each set of valuations $v_{-i}$ of the other players the menu presented to $i$ is identical in both $A$ and $A'$. Therefore it suffices to prove that the taxation complexity of player $i$ in $A'$ is at most ${cc(A')}$. This implies that the taxation complexity of $A$ is $cc(A)+1$.

With this in mind, let $M^i=\{\mathcal M|\text{$\exists (v_1,\ldots,v_{i-1},v_{i+1},\ldots v_n)$ s.t. $\mathcal M$ is presented to $i$}\}$ be the set of all possible menus that might be presented to player $i$. For each $\mathcal M\in M^i$, let $v^{\mathcal M}_i$ be the valuation in which for every $S$ we have that $v^{\mathcal M}_i(S)=\mathcal M(S)$. Observe that each $v^{\mathcal M}_i$ is a valid valuation function: by Proposition \ref{proposition-menu-monotonicity}, $\mathcal M$ is monotone and normalized, thus $v^{\mathcal M}_i$ is monotone and normalized as well. Also notice that for each $S$, $v^{\mathcal M}_i(S)<\infty$ since $\mathcal M(S)<\infty$. In addition, for every $\mathcal M\in M^i$ choose an arbitrary set of valuations of the other players $v^{\mathcal M}_{-i}=(v^\mathcal M_1,\ldots, v^\mathcal M_{i-1},v^\mathcal M_{i+1},\ldots, v^\mathcal M_n)$ such that if the players' valuations are $v^{\mathcal M}_{-i}$ the menu that player $i$ is presented with is $\mathcal M$.

Now we get to the heart of the proof. We will show that for each $\mathcal M, \mathcal M'\in M^i$, $\mathcal M\neq \mathcal M'$, the transcript of $A'$ in the instance $(v^{\mathcal M}_i,v^{\mathcal M}_{-i})$ differs from the transcript of $A'$ in the instance $(v^{\mathcal M'}_i,v^{\mathcal M'}_{-i})$. Recall that the communication complexity of $A'$ is $cc(A)+1$, thus there are at most $2^{cc(A)+1}$ different transcripts. The bound on the taxation complexity of $A'$ will then follow since every instance of the form $(v^{\mathcal M}_i,v^{\mathcal M}_{-i})$ corresponds to exactly one menu $\mathcal M\in M^i$.

\begin{claim}
For every $\mathcal M, \mathcal M'\in M^i$, $\mathcal M\neq \mathcal M'$, the transcript of $A'$ in the instance $(v^{\mathcal M}_i,v^{\mathcal M}_{-i})$ differs from the transcript of $A'$ in the instance $(v^{\mathcal M'}_i,v^{\mathcal M'}_{-i})$.
\end{claim}
\begin{proof}
Assume towards contradiction that there are $\mathcal M, \mathcal M'\in M^i$, $\mathcal M\neq \mathcal M'$ such that the transcript of the instance $(v^{\mathcal M}_{-i}, \mathcal M)$ is identical to the transcript of the instance $(v^{\mathcal M'}_{-i}, \mathcal M')$. Using standard fooling-set arguments (e.g., \cite{KN97}), this implies that the transcripts of $(v^{\mathcal M}_{-i}, \mathcal M')$ and $(v^{\mathcal M'}_{-i}, \mathcal M)$ are identical as well. We will show that this is not the case and reach a contradiction. Towards this end, observe that the last bit that player $i$ sends in both instances is by construction $0$ (since $v_i^{\mathcal M}(S)=\mathcal M(S)$ and $v_i^{\mathcal M'}(S)=\mathcal M'(S)$ for every bundle $S$, so the profit in both instances is $0$). However, we will show that in either $(v^{\mathcal M}_{i}, v^{\mathcal M'}_{-i})$ or $(v^{\mathcal M'}_{i}, v^{\mathcal M}_{-i})$ the last bit that player $i$ sends is $1$. In particular we get a different transcript, which is a contradiction.

To see that in one of the instances $(v^{\mathcal M}_{i}, v^{\mathcal M'}_{-i})$ and $(v^{\mathcal M'}_{i}, v^{\mathcal M}_{-i})$ the last bit that is communicated is $1$, notice that since $\mathcal M\neq \mathcal M'$, there must be a bundle $S$ such that $\mathcal M(S)\neq \mathcal M'(S)$. Assume without loss of generality that $\mathcal M(S)> \mathcal M'(S)$. Thus in the instance $(v^{\mathcal M'}_{i}, v^{\mathcal M}_{-i})$ the profit of player $1$ for the bundle $S$ is $v_i^\mathcal M(S)-\mathcal M'(S)>0$. By the taxation principle, player $i$ must win a bundle with at least that (positive) profit. Thus the last bit that player $i$ communicates is $1$, which gives us the desired contradiction.
\end{proof}

This concludes the proof of Theorem \ref{thm-warm-up}.
\end{proof}

\paragraph{Tightness.} To see that the Theorem \ref{thm-warm-up} is essentially tight, we present a mechanism with taxation complexity very close to the communication complexity. Consider the following truthful mechanism for combinatorial auctions with two players, Alice and Bob. Let $M^{Bob}=\{\mathcal M_1, \ldots, \mathcal M_{2^c}\}$ be a set of $2^c$ menus, where for each $\mathcal M_i\in M^{Bob}$ we have that $\mathcal M_i(\{a\})=i$, $\mathcal M_i(\emptyset)=0$, and $\mathcal M_i(S)=\infty$ for every $S\neq \emptyset,\{a\}$. Let $t$ be Alice's value for item $a$ rounded to the nearest integer in $1,2,\ldots, 2^c$. Alice now sends $t$ using $c$ bits of communication. If Bob's value for item $a$ is at least $t$, Bob sends the bit $1$, receives $a$, and pays $t$. Otherwise, he sends the bit $0$, receives no items at all, and pays nothing. Alice always receives the empty bundle. The mechanism is clearly truthful, its communication complexity $c+1$, and its taxation complexity is $c$.

\subsection{Bounding the Taxation Complexity: The Full Result}\label{subsec-main}

We now significantly strengthen the results of Subsection \ref{subsec-warmup}. In particular, we give bounds on the taxation complexity also for mechanisms that are truthful for restricted classes of valuations (subadditive, XOS, and submodular). We will show that:

\begin{theorem}\label{thm-indexing-complexity}
Let $\mathcal V$ be some class of valuations. Fix a mechanism $A$ for combinatorial auctions that is truthful when the valuations of the players are in $\mathcal V$. Then:
\begin{enumerate}[noitemsep]
\item If $\mathcal V$ is the set of all normalized and monotone valuations then $tax(A)\leq cc(A)$.
\item If $\mathcal V$ is the set of subadditive valuations then $tax(A)\leq cc(A)$.
\item If $\mathcal V$ is the set of XOS valuations then $tax(A)\leq m\cdot (cc(A)+1)$.
\item If $\mathcal V$ is the set of submodular valuations then $tax(A)\leq d\cdot m\cdot (cc(A)+1)$, where $d=|\{\mathcal M(S)\}_{\mathcal M\in M^i, S}|$ is the total number of distinct prices that appear in some menu.
\end{enumerate}
\end{theorem}
The proof of Theorem \ref{thm-indexing-complexity} is postponed to Appendix \ref{app-proof-tax-complexity} due to lack of space. We note that the vast majority of the algorithms in the literature are truthful for general valuations (e.g., maximal in range algorithms, posted prices mechanisms). The restrictions on the valuations are typically used only for the performance analysis. The taxation complexity of those mechanisms is therefore at most their communication complexity. We also remark that the bound on the taxation complexity for submodular valuations depends on the number of possible prices. A natural open question is:
\begin{open}
Let $A$ be a mechanism for combinatorial auctions that is truthful for submodular valuations. Is $tax(A)\leq poly(cc(A),m)$?
\end{open}
More generally, we have already mentioned two domains in which there is a truthful mechanism $A$ with $tax(A)>>cc(A)$. The first was combinatorial auctions with gross substitute valuations. For combinatorial auctions with general valuations we mentioned a randomized truthful in expectation mechanism, but this mechanism (as well as all truthful in expectation mechanisms) can be seen as a deterministic one by letting the range be the set of all possible distributions over allocations and letting the value of a player for a distribution be the expected value of the bundle he receives in that distribution. This leads us to the following question:
\begin{open}
Characterize the set of domains in which for every truthful mechanism $A$ we have that $tax(A)\leq poly(cc(A),m)$.
\end{open}

\section{Menu Reconstruction: $cc(A)\leq poly(tax(A), price(A), tie(A),m,n)$}\label{sec-menu-construction}

Our goal in this section is to provide a characterization of the communication complexity of truthful mechanisms. Our first task is to develop a low communication protocol that lets $n-1$ players find the menu they present to the remaining player.


\begin{theorem}[The Menu Reconstruction Theorem]\label{thm-construct-menu}
Fix a truthful mechanism $A$. Denote by $v_{-i}$ the valuation profile of all players except $i$. The communication complexity of finding the index of the menu presented to $i$ by $v_{-i}$ is $poly(tax(A),price(A),m,n)$. 
\end{theorem}
In the statement of the theorem, we denote by $price(A)$ is the communication complexity of the following $(n-1)$-player problem: fix some truthful mechanism $A$, player $i$, and bundle $S$. The input of each player $i'\neq i$ is a valuation $v_{i'}$. Let $\mathcal M$ be the menu that is presented to $i$ in $A$ when the valuations are $v_{-i}$. $price(A)$ is the communication complexity of computing $\mathcal M(S)$.

Due to lack of space, we bring the proof of the theorem in Appendix \ref{sec-thm-construct-menu}. We note that although for concreteness the menu reconstruction theorem is proved for combinatorial auctions, it actually applies to \emph{any} domain. That is, fix any truthful $n$-player mechanism $A$ whose range is a set of alternatives $\mathcal A$. Then, the menu presented to any player can be found using $poly(tax(A),price(A),\log |\mathcal A|,n)$ bits (again, $price(A)$ is the communication complexity of finding the price of an alternative $S\in \mathcal A$).

We now want to express the communication complexity of menu reconstruction in terms of $cc(A)$. Proposition \ref{proposition-price-of-bundle} shows that if $A$ is truthful for additive valuations, then $price(A)\leq cc(A)$. Since Theorem \ref{thm-indexing-complexity} gives us that for general valuations $tax(A)\leq cc(A)$, we get that finding the full menu is not much harder than determining the allocation and the prices of the winning bundles:
\begin{corollary}
Fix a mechanism $A$ that is truthful for general valuations. Denote by $v_{-i}$ the valuation profile of all players except $i$. The communication complexity of finding the index of the menu presented to $i$ by $v_{-i}$ is $poly(cc(A),m,n)$. 
\end{corollary}
Similar bounds also hold for other valuation classes, by using Theorem \ref{thm-indexing-complexity} appropriately.

More importantly, we can now characterize the communication complexity of truthful mechanisms. Let $tie(A)$ be the communication complexity of determining the allocation of $A$ when the valuation of each player $i$ is $v_i$ and all players know the menu $\mathcal M_i$ player $i$ is presented with by $v_{-i}$\footnote{The notation $tie(A)$ hints that this is a question about tie breaking: each player must be allocated a profit maximizing bundle and the set of profit maximizing bundles can be computed without additional communication as it depends on $\mathcal M_i$ and $v_i$ only. However, deciding which specific bundle in the set the player is allocated might depend also on the valuations of the other players and might require extra communication.}. Notice that $tie(A)\leq cc(A)$, since we can always run $A$ and ignore the extra information about the $\mathcal M_i$'s. This gives a characterization of the communication complexity of mechanisms that are truthful for general valuations (again, similar bounds hold by applying other parts of Theorem \ref{thm-indexing-complexity}):
\begin{theorem}[characterization of the communication complexity of truthful mechanisms]\label{thm-characterization}
Fix a mechanism $A$ that is truthful for general valuations. Then:
$$
\frac {tax(A)+price(A)+tie(A)} 3 \leq cc(A) \leq poly(tax(A),price(A),tie(A),m,n) 
$$
\end{theorem}
\begin{proof}
We first prove the LHS. We always have that $tie(A)\leq cc(A)$. $A$ is truthful for general valuations and hence it is also truthful for additive valuations, therefore by Proposition \ref{proposition-price-of-bundle}, $price(A)\leq cc(A)$. To finish this part, observe that by Theorem \ref{thm-indexing-complexity}, $tax(A)\leq cc(A)$.

The RHS is obtained by applying the menu reconstruction theorem $n$ times, once for each player. Then, we need additional $tie(A)$ communication bits to determine the final allocation.
\end{proof}

In Subsection \ref{sec-tight-char} we show that our characterization is tight in the sense that if we drop at least one of the three main terms ($tax(A),price(A), tie(A)$) then the gap between the LHS and the RHS might be exponential. For instance, we have already mentioned an example of a truthful mechanism $A$ with $tax(A)=1$ (and thus $price(A)=0$) in which $cc(A)=exp(m)$. We also provide examples with similar gaps when dropping $price(A)$ and $tax(A)$.

\subsubsection*{Acknowledgments}

I thank Sigal Oren for valuable advice and countless discussions during the work on this paper. I thank Hu Fu and Omri Weinstein for comments on an earlier draft. I am also grateful to the participants of the Hebrew University's AGT seminar for their very helpful feedback when this work was in a preliminary stage.

\bibliographystyle{plain}
\bibliography{bib}

\appendix

\section{Formalities}\label{sec-preliminaries}

\subsection{Combinatorial Auctions}

In a combinatorial auction there is a set $M$ items ($|M|=m$) and a set of $N$ players ($|N|=n$). Each player $i$ has a valuation function $v_i:2^M\rightarrow \mathbb R$ that denotes the value of player $i$ for every possible subset of the items. We assume that the valuation functions are monotone (for all $S\subseteq T$, $v_i(S)\geq v_i(T)$) and normalized ($v_i(\emptyset)=0$). The output is an allocation of the items $(S_1,\ldots, S_n)$.

We will sometimes consider additional restrictions on the valuations:
\begin{enumerate}
\item \textbf{Additive:} a valuation $v$ is \emph{additive} if for every bundle $S$ we have that $v(S)=\Sigma_{j\in S}v(\{j\})$. 
\item \textbf{Submodular:} a valuation $v$ is \emph{submodular} if for every two bundles $S$ and $T$ it holds that $v(S)+v(T)\geq v(S\cup T)+v(S\cap T)$.
\item \textbf{XOS: }a valuation $v$ is \emph{XOS} if there exist additive valuations $a_1,\ldots, a_t$ such that for every bundle $S$, $v(S)=\max_r a_r(S)$. Each $a_r$ is a \emph{clause} of $v$. If $a\in\arg\max_r a_r(S)$ then $a$ is a \emph{maximizing clause} of $S$ and $a(j)$ is the \emph{supporting price} of item $j$ in this maximizing clause.
\item \textbf{Subadditive:} a valuation $v$ is \emph{subadditive} if for every two bundles $S$ and $T$ it holds that $v(S)+v(T)\geq v(S\cup T)$.
\end{enumerate}

It is known \cite{LLN01} that each class defined above contains its predecessors, and that these containments are strict. 

\subsection{Ex-Post Nash Equilibrium and Dominant Strategies}

Consider a general iterative (not necessarily direct) mechanism for $n$ players. Denote the type space of player $i$ by $\mathcal T_i$ and the set of possible alternatives by $\mathcal A$. The mechanism works in $r$ rounds (we restrict ourselves to finite mechanisms with $r<\infty$), where in each round $r'$ each player $i$ observes the actions chosen by all players (including himself) in the previous rounds and chooses an action for this round from a set $\mathcal X_i$. Let $\mathcal H_{r'}$ denote the set of all possible actions played by all players in rounds $1,\ldots, r'$, i.e., the history of the game in rounds $1,\ldots,r'$. 

The output of the game is determined by $A:\mathcal H_r\rightarrow \mathcal A$. Each player $i$ has a valuation function $v_i:\mathcal T_i\times \mathcal A\rightarrow \mathbb R$. There are also payment functions $p_i:\mathcal H_r\rightarrow \mathbb R$. In this paper we assume that the utility function $u_i$ of player $i$ is quasi linear: $v_i(t_i,  A(h_r))-p_i(h_r)$, where $t_i$ is the type of player $i$, and $h_r$ is the history of the game. Thus, a strategy $s_i$ is simply a set of functions $s^{1}_i,\ldots ,s^{r}_i$ where each $s^{r'}_i:T_i\times \mathcal H_{r'-1}\rightarrow X_i$ determines the action of player $i$ in round $r'$ given his type and the history. We will sometime use $h(s_1,\ldots, s_n)$ to denote the history of the game where each player $i$ plays according to the strategy $s_i$. We define two notions of equilibria:

\begin{definition}[dominant strategy equilibrium]
A strategy $s_i$ is a \emph{dominant strategy} for player $i$ if for every $t_i$ and $s_{-i}$, $u_i(t_i,h(s_i, s_{-i}))\geq u_i(t_i,h(s'_i, s_{-i}))$ for all $s'_i$. Strategies $(s_1,\ldots, s_n)$ constitute a \emph{dominant strategy equilibrium} if each $s_i$ is dominant.
\end{definition}

\begin{definition}[ex-post Nash equilibrium]
For each player $i$, let $s_i$ be a function that takes player $i$'s type and outputs a strategy. $(s_1,\ldots, s_n)$ constitute an ex-post Nash equilibrium if for every player $i$ with type $t_i$ and every type profile $t_{-i}$ of the other players it holds that $u_i(t_i,h(s_i(t_i), s_{-i}(t_{-i})))\geq u_i(t_i,h(s'_i(t'_i), s_{-i}(t_{-i})))$ for all $t'_i\in \mathcal T_i$.
\end{definition}

In this paper we use the term \emph{truthful} to denote mechanisms that reach an ex-post Nash equilibrium. Observe that every dominant strategy equilibrium is also an ex-post Nash equilibrium, but the other direction is not true: consider a second price auction with two players where the value of each player $i$ for the item is $v_i$. As usual, each player $i$ submits a bid $b_i$, the item goes to the player with the highest bid who pays the bid of the other player. If the players submit their bids simultaneously, then of course setting $b_i=v_i$ is a dominant strategy for each of the players. However, consider an iterative game where player $1$ bids first and player $2$ bids after he sees $b_1$. In this game, setting $b_1=v_1$ is no longer a dominant strategy for player $1$. To see that, consider the following strategy of player $2$: if $b_1=1$ player $2$ sets $b_2=0.99$ and otherwise $b_2=1$. Notice that given this strategy if $v_1>1$ player $1$ is better off bidding $b_1=1$ rather than $b_1=v_1$. However, the set of strategies where each player $i$ bids $b_i=v_i$ does constitute an ex-post Nash equilibrium.

This paper follows the usual formulation of combinatorial auctions as a game. We identify between the type space and the valuation function, so $v_i$ is the private information of each player $i$. The set of alternatives $\mathcal A$ is the set of all possible allocations.

\subsection{Menus, the Taxation Principle, and Taxation Complexity}

A key component of this paper is the taxation principle. The taxation principle holds for every domain, but for convenience we specialize it here for combinatorial auctions:
\begin{proposition}[taxation principle]
Consider some mechanism $A$ for combinatorial auctions and let $(s_1,\ldots,s_n)$ be an ex-post Nash equilibrium in this mechanism. Fix some player $i$ and $v_{-i}$. Then, for every bundle $S$ there is a price $p_S$ such that for every $v_i$, when the players play according to $(s_i(v_i),s_{-i}(v_{-i}))$ and player $i$ wins $S$, the payment of player $i$ is $p_S$.
\end{proposition}
\begin{proof}
Consider $v_i$ and $v'_i$ such that player $i$ is allocated $S$ when the players play according to both  $(s_i(v_i),s_{-i}(v_{-i}))$ and $(s_i(v'_i),s_{-i}(v_{-i}))$, and charged $p$ and $p'$, respectively. If $p\neq p'$, suppose without loss of generality that $p>p'$. Notice that the profit of player $i$ with valuation $v_i$ is $v_i(S)-p$ when playing according to $s_i(v_i)$. However, if player $i$ plays according to $s_i(v'_i)$ he still wins the bundle $S$ but pays only $p'$, so his profit is $v_i(S)-p'>v_i(S)-p$. A contradiction to the assumption that the profile $(s_1,\ldots, s_n)$ is an ex-post Nash equilibrium.
\end{proof}

We set $p_S=\infty$ if for some $S$ there is no $v_i$ such that when the players play according to $(s_i(v_i),s_{-i}(v_{-i}))$ and player $i$ is allocated $S$. Note that the taxation principle gives a natural interpretation to any ex-post Nash equilibrium: each player $i$ is presented with a menu $\mathcal M:2^M\rightarrow \mathbb R\cup \{\infty\}$ that depends only on $v_{-i}$. In equilibrium, player $i$ is assigned a bundle that maximizes his profit $\arg\max_S v_i(S)-\mathcal M(S)$. We will say that $\mathcal M$ is \emph{presented} to player $i$ by $v_{-i}$. 

\begin{proposition}[menu monotonicity]\label{proposition-menu-monotonicity}
Consider some truthful mechanism $A$ for combinatorial auctions and let $(s_1,\ldots,s_n)$ be an ex-post Nash equilibrium in this mechanism. Fix some player $i$ and $v_{-i}$. Let $\mathcal M$ be the menu presented to player $i$ by $v_{-i}$. Then, without loss of generality we can assume that $\mathcal M$ is monotone: for every $S\subseteq T$, $\mathcal M(T)\geq \mathcal M(S)$. Furthermore, we can assume without loss of generality that $A$ is normalized: $\mathcal M(\emptyset)=0$.
\end{proposition}
\begin{proof}
We first prove that $\mathcal M$ is monotone. Suppose that $\mathcal M(T)<\mathcal M(S)$. Then, for every $v_i$, $v_i(S)-\mathcal M(S)<v_i(T)-\mathcal M(T)$ since by the monotonicity of the valuations $v(T)\geq v(S)$. In other words, player $i$ never wins the bundle $T$. Therefore, in this case setting $\mathcal M(S)=\mathcal M(T)$ is consistent with the social choice function: the profit from $S$ is at least the profit from $T$, so we may have only increased the set of most profitable bundles, and we can still assume by tie-breaking that $T$ is never chosen.

As for normalization, if $\mathcal M(\emptyset)\neq 0$, define a new menu $\mathcal M'$ with $\mathcal M'(S)=\mathcal M(S)-\mathcal M'(\emptyset)$. Notice that we may assume that player $i$ is presented with $\mathcal M'$ and not with $\mathcal M$ since shifting all prices by a constant does not change the set of profit maximizing bundles. 
\end{proof}


Denote by $M^i=\{\mathcal M_{v_{-i}}\}_{v_{-i}}$ the set of menus that might be presented to $i$. Denote by $tax(A)$ the \emph{taxation complexity} of a truthful mechanism $A$ -- the number of bits needed to represent an index of a specific menu among the set of menus that may be presented to a player. That is, $tax(A)=\max_i\log |M^i|$. We sometimes also refer to $\log |M^i|$ as the \emph{taxation complexity of player $i$}.




\subsection{Computational Models}

This paper considers three ways in which the players communicate, which correspond to three ways of accessing the valuation functions:
\begin{itemize}
\item {\emph{Value queries:}} Each valuation $v$ is represented by a black box that can answer only the following question: given $S$, what is $v(S)$?

\item {\emph{Demand queries:}} Each valuation $v$ is represented by a black box that can answer only the following question: given prices per item $p_1,\ldots, p_m$, what is a profit maximizing bundle $S\in\arg \max v(T)-\Sigma_{j\in T}p_j$? If there are several bundles that maximize the profit we use a fixed tie breaking rule to determine which bundle will be returned (say, the lexicographically first one). For simplicity we assume that the value $v(T)$ is also returned.

\item {\emph{General communication:}} This is the usual number-in-hand communication model (see \cite{KN97}) where we assume that the input of player $i$ is his valuation $v_i$. At each round, each player $i$ decides which bits he sends based on $v_i$ and the bits sent by all players in the previous rounds.
\end{itemize}

If the players communicate only by answering demand or value queries, then the complexity of the mechanism is the largest number of queries that the mechanism makes over all inputs. In the general communication model the complexity of the mechanism is the largest number of bits that the players send. In all models, the maximum is taken over all possible inputs.

Notice that each way corresponds to a different restriction on the action space in the game theoretic formulation. I.e., if the valuations can only be accessed by value queries, then the action space consists only of answering a value query.

\subsubsection{Representation of Numbers}\label{subsec-representation}

We would like to explicitly discuss the delicate issue of representing numbers, which is step-sided in many of the previous works on algorithmic mechanism design and communication complexity. In general, we follow the standard formulation (see, e.g., \cite{BN07,N14,FS09,BBNS08}) and assume that all numbers (in particular the values and prices) are limited to a certain precision, i.e., are represented by some number of bits $k$. We limit our attention to protocols that take a precision parameter $k$ with the following property: let $P_k$ and $P_{k'}$ be the same protocol except that the precision parameter is either $k$ or $k'$, where $k'<k$. We require that when the input can be represented by at most $k'$ bits of precision, the output (allocation and prices) is identical in $P_k$ and $P_{k'}$. Notice that the complexity of the mechanism should also take $k$ into account, although we will mostly think about the precision $k$ as fixed. To the best of our knowledge, all protocols that were considered in the Algorithmic Mechanism Design literature have this property.

We sometimes let $B$ be the maximum price that may appear in a menu. Notice that since all prices are represented by a finite number of bits, $B$ is well defined. In some proofs (e.g., proof of Theorem \ref{thm-indexing-complexity}) we use valuations in which some of the values of the bundles are a function of $B$ (e.g., $v(S)=2\cdot B$). One issue is that the number of bits needed to represent these values is bigger than $k$. However, since all numbers we use are not larger than $2^m\cdot B$ (and usually much smaller), whenever we analyze valuations that use $k$ bits of precision, we use protocols that allow the representation of $m\cdot k$ bits of precision, which allow representation of values such that $2^m\cdot B$. Notice that if the communication complexity depends polynomially on the precision parameter $k$, the overall communication blow-up due to the use of increased precision is only $poly(m)$.

\subsection{Chernoff Bounds}

We will need the following version of the Chernoff bounds:

\begin{proposition}[Chernoff bounds]\label{claim-chernoff}
Let $X_1,...X_n$ be independent random variables that take values in $\{0,1\}$, such that for all $i$, $\Pr[X_i=1]=p$ for some $p$. Then, the following holds, for $0\leq \epsilon\leq 1$:
\begin{enumerate}
\item $\Pr[\Sigma_iX_i>(1+\epsilon)pn]\leq e^{-pn\epsilon^2/3}$
\item $\Pr[\Sigma_iX_i<(1-\epsilon)pn]\leq e^{-pn\epsilon^2/2}$
\end{enumerate}
\end{proposition}

\section{Characterizing the Query Complexity of Truthful Mechanisms}\label{sec-query-complexity}

In this section we handle mechanisms that can access the valuations only with restricted type of queries. For value queries, we show that the menu complexity characterizes the query complexity in exactly the same way that the taxation complexity characterizes the communication complexity. For demand queries, we show an impossibility result: a menu reconstruction theorem for demand queries does not exist. Nevertheless, we are able to characterize the \emph{structure} of the menu in that case, in a way that sheds lights on the interplay between the number of queries and the semantics of the mechanism. We will also see that this characterization is useful for proving impossibility results on the power of truthful computationally efficient mechanisms.

In both settings we develop our results for mechanisms that are truthful for general valuations. However, extending the results also for other valuations classes should be possible by applying the ideas similar to the proof of Theorem \ref{thm-indexing-complexity}.



\subsection{Value Queries}\label{sec-value}

We now consider mechanisms in which the players' valuations are represented by black boxes that can only answer value queries. In this setting, the analogue of taxation complexity will be the \emph{menu complexity}. That is, let $\mathcal M^i$ be the set of menus that might be presented to player $i$ in some truthful mechanism $A$. Roughly speaking, the menu complexity of a menu $\mathcal M$ that is presented to player $i$ measures the number of bundles $S$ that $i$ might win. The exact definition is a bit more delicate: suppose that the menu $\mathcal M$ presented to player $i$ is identically zero. If player $i$'s valuation is identically zero as well he might win any bundle in case we have some strange tie-breaking rule that depends on the valuations of the other players. To overcome this tie-breaking issue we essentially need to consider only bundles $S$ that player $i$ might win when his valuation is \emph{strictly} monotone. We note that in this section we assume without loss of generality that the menu is monotone in the sense of Proposition \ref{proposition-menu-monotonicity}.

\begin{definition}[menu complexity]
Let $A$ be some truthful mechanism. Consider some menu $\mathcal M$ that is presented to player $i$ when the valuations of the other players are $v_{-i}=(v_1,\ldots, v_{i-1},v_{i+1},\ldots, v_n)$. We say that bundle $S\neq M$ is \emph{in $\mathcal M$} if for every $T\supset S$ we have that $\mathcal M(S)<\mathcal M(T)$. The grand bundle $M$ is in the menu if $\mathcal M(M)<\infty$.

The \emph{menu complexity} of $\mathcal M$, denoted $mc(\mathcal M)$ is the number of bundles that are in a menu $\mathcal M$. The \emph{menu complexity of $A$} is $mc(A)=\max_i\max_{\mathcal M\in M^i}mc(\mathcal M)$.
\end{definition}

It is not hard to see that equivalently we could have defined the menu complexity of a menu $\mathcal M$ to be:
$$
mc(\mathcal M)=\{S|\exists v\in \mathcal V' \text{ such that $A$ awards $i$ the bundle $S$ in the instance } (v,v_{-i})\}
$$
where $\mathcal V'$ is the set of strictly monotone valuations.

Denote by $val(A)$ the maximum number of value queries that the most efficient implementation of $A$ makes on any input. 

\subsubsection{Bounding the Query Complexity: $mc(A)\leq val(A)$}

Our first theorem shows that the menu complexity is at most the query complexity:
\begin{theorem}\label{thm-value}
Let $A$ be a mechanism that is truthful for general valuations and uses only value queries. Then, $mc(A)\leq val(A)+2$.
\end{theorem}
A weaker variant of this theorem can be obtained as a special case of Theorem \ref{thm-demand} that we later prove by setting $\alpha=0$ and $\beta=val(A)$. In fact, Theorem \ref{thm-value} is also implicit in \cite{D11} and \cite[Theorem 11]{N14}. We bring the explicit proof here:

\begin{proof}(of Theorem \ref{thm-value})
Fix a truthful mechanism $A$, some player $i$, and valuations $v_{-i}=(v_1,\ldots, v_{i-1},v_{i+1},\ldots, v_n)$ of all other players. Let $\mathcal M$ be the menu that is presented to player $i$ in $A$ by $v_{-i}$. Let $v_i$ be the valuation that is defined by $v_i(S)=\mathcal M(S)$ for every $S$ with $\mathcal M(S)<\infty$. For every $S$ with $\mathcal M(S)=\infty$ we set $v_i(S)$ to be an arbitrary value (which is strictly bigger than any finite price in $\mathcal M$). We will show that in the instance $(v_i,v_{-i})$ any implementation of the truthful mechanism $A$ makes at least $mc(\mathcal M)-2$ value queries, and the theorem will follow.

Let $S'$ be the bundle that player $i$ is awarded in the instance $(v_i,v_{-i})$. Consider some bundle $S\neq S'$, $S\neq \emptyset$ that is in the menu $\mathcal M$ (there are at most $mc(A)-2$ such bundles). We will show that in the instance $(v_i,v_{-i})$ the mechanism $A$ queries $v_i(S)$ and finish the proof. Suppose not. Let $v'_i$ be the valuation that is identical to $v_i$ except that $v'_i(S)=v_i(S)+\epsilon$, for some $\epsilon>0$ that preserves the monotonicity of $v'_i$. Such $\epsilon>0$ exists since $S$ is in the menu $\mathcal M$. Notice that if $A$ does not query $v_i(S)$, it cannot distinguish between the valuations $v_i$ and $v'_i$. Therefore player $i$ receives the bundle $S'$ also in the instance $(v'_i,v_{-i})$. However, the profit from the bundle $S$ is $v'_i(S)-\mathcal M(S)=v_i(S)+\epsilon-\mathcal M(S)=\epsilon>0$ whereas the profit from $S'$ is $v'_i(S)-\mathcal M(S)=v_i(S)-\mathcal M(S)=0$. I.e., player $i$ is not awarded his most profitable bundle, a contradiction to the taxation principle.
%
\end{proof}

\paragraph{Tightness.} To see that Theorem \ref{thm-value} is essentially tight (i.e., there is a mechanism with menu complexity very close to the number of value queries it makes), consider the following truthful mechanism for combinatorial auctions with two players, Alice and Bob. Fix some set $T=\{T_1,\ldots, T_c\}$ of bundles, none of them is the empty bundle. Let $t$ be Alice's value for item $a$ rounded to the nearest integer in $1,2,\ldots, c$ (determining $t$ can be done by making one value query). Set the price of each bundle $S\in T, S\neq T_t$ to $|S|$. Set the price of $T_t$ to $|T_t|+\frac 1 2$. Bob receives the bundle from the set $T$ that maximizes his profit according to these prices (breaking ties in some consistent way) and pays the appropriate price, unless it has a negative profit, in which case he is allocated no items at all. Alice never receives any items.

The mechanism is clearly truthful, uses $c+1$ value queries (determining the profit-maximizing bundle requires $c$ value queries, one for each bundle in $T$). Its menu complexity is $c$.

\subsubsection{Menu Reconstruction and Characterization}

We now provide a menu reconstruction theorem that uses only value queries. Let $price^{val}(A)$ be defined similarly to the definition of $price(A)$ in Section \ref{sec-menu-construction} but with respect to value queries.
\begin{theorem}[a menu reconstruction theorem for value queries]\label{thm-value-reconstruct}
Fix a truthful mechanism $A$ that makes only value queries. Let $v_{-i}$ be the valuations of all players except $i$. The index of the menu presented to $i$ by $v_{-i}$ can be found by making $poly(mc(\mathcal A),price^{val}(A),m)$ value queries. 
\end{theorem}
\begin{proof}
A valuation $v$ is $k$-useless if there exist (not necessarily unique) sets $K_1,\ldots, K_k$ such that $v$ can be described as follows:
$$
v(S)= \begin{cases} 
 0 &  \exists t: S\subseteq K_t, \\
1 & \text{otherwise.} 
\end{cases}
$$
We call each $K_t$ a \emph{useless} bundle of $v$. 

Suppose that we are given a valuation $v$ on $m$ items that can be accessed via value queries only. We are guaranteed that $v$ is $k$-useless, but we do not know $K_1,\ldots, K_k$. We would like to \emph{learn} the $k$-useless valuation $v$, that is, obtain an algorithm that makes value queries only and finds $K_1,\ldots, K_k$. The next claim shows that an algorithm for learning $k$-useless valuations yields a menu reconstruction theorem:

\begin{claim}
Let $q_k$ be the query complexity of learning $k$-useless valuations. The index of the menu presented to $i$ by $v_{-i}$ can be found by making $poly(mc(\mathcal A),price^{val}(A),q_{mc(A)},m)$ value queries. 
\end{claim}
\begin{proof}
Let $\mathcal M$ be the menu presented by $v_{-i}$ to player $i$. For every price $p$, define a valuation $v^{p}$:
$$
v^p(S)= \begin{cases} 
 0 &  \mathcal M(S)\leq p, \\
1 & \mathcal M(S)> p. 
\end{cases}
$$
Notice that we can easily simulate a value query $v^p(S)$ using $price^{val}(A)$ value queries: compute $\mathcal M(S)$ by making $price^{val}(A)$ value queries, and determine whether $v^p(S)$ is $0$ or $1$ accordingly.

The motivation to the definition of $v^p$ comes from the following observation: let $S$ is a bundle that is in $\mathcal M$. Let $p=\mathcal M(S)$. Then $S$ is one of the useless bundles of $v^p$. This is simply because $v^p(S)=0$ by definition, and since for every $j\notin S$ we have that $\mathcal M(S+\{j\})>p$ and thus $v^p(S+\{j\})=1$, precisely because $S$ is in the menu $\mathcal M$. Similarly, if $S'$ is a useless bundle of $v^p$, then $S'$ is in $\mathcal M$. This shows that for every $p$, $v^p$ is a $mc(A)$-useless valuation.

It is also easy to see that if there is some bundle $S$ with $\mathcal M(S)=p$ then there is some bundle $S'$ that is in $\mathcal M$ with $\mathcal M(S')=p$: start with $S$, and check if there is some item $j\notin S$ such that $\mathcal M(S)=\mathcal M(S+\{j\})$. If not, then $S$ is in the menu, else, repeat the process but now with the bundle $S+\{j\}$ instead of $S$. The process stops after at most $|M-S|$ additions of items with a bundle that is in $\mathcal M$ and has price $p$, simply because we run out of items to add.

We can now run the following natural algorithm to find all bundles that are in $\mathcal M$. First, consider $v^0$ and use $q_k$ value queries to $v^0$ (each costs $price^{val}(A)$ ``real'' value queries to $v_{-i}$) to determine the useless bundles of $v^0$. Notice that the set of useless bundles is not empty, since $\mathcal M(\emptyset)=0$. We have already observed above that the set of useless bundles of $v^0$ contains only bundles that are in $\mathcal M$, and furthermore contains all bundles that are in $\mathcal M$ with price $0$. 

Now we find the minimal price $p$ of some bundle in $\mathcal M$ that is bigger than $0$. Denote the value query complexity by $r$. Define $v^p$ and as before use $q_k\cdot price^{val}(A)$ value queries to find all the sets that are in $\mathcal M$ and have price $p$. The process will stop after $mc(A)$ iterations, since there are at most $mc(A)$ distinct prices in $\mathcal M$ (recall that if there is a bundle $S$ with price $\mathcal M(S)=p$ then there is a bundle in $\mathcal M$ with price $p$). At the end of the process we have found all $mc(A)$ bundles that are in the menu. Use $price^{val}(A)$ to determine the price of each of them.

The total number of value queries is therefore $mc(A)\cdot (r+q_k\cdot price^{val}(A))+mc(A)\cdot price^{val}(A)$. To finish the proof we prove that $r\leq q_k$.

\begin{claim}
Any algorithm that learns a $k$-useless valuation $v^p$ that is obtained from $v$ as above must query at least one set $S$ with $\mathcal M(S)=\min_{S':\mathcal M(S')>p}\mathcal M(S')=p'$.
\end{claim}
\begin{proof}
Let $\mathcal M'$ be the menu that is identical to $\mathcal M$ except for bundle each $S$ with $\mathcal M(S)=p'$ for which we have that $\mathcal M'(S)=p$. Notice that $\mathcal M'$ is monotone since $\mathcal M$ is monotone and since there is no bundle $S$ with $p<\mathcal M(S)<p'$. Notice that the algorithm for learning a $k$-useless valuation will not notice the difference between $\mathcal M$ and $\mathcal M'$ (since we only changed the prices of bundles that were not queried), thus the set of useless bundles it returns is the same. However, the valuation $v'^p$ obtained from $\mathcal M'$ has at least one additional useless bundle (recall that if there is a bundle with price $\mathcal M(S)$ then there is a bundle in the menu with the same price). This is a contradiction to the correctness of the algorithm that learns $k$-useless valuations.
\end{proof}

I.e., we can compute $\min_{S':\mathcal M(S')>p}\mathcal M(S')$ just by running the algorithm for learning a $k$-useless valuation on $v^p$, and take the minimal value that is strictly bigger than $p$ that was encountered when the valuation $v$ was queried. Thus, we can conclude that the index of the menu presented to player $i$ can be found with at most $mc(A)\cdot q_k\cdot price^{val}(A)+mc(A)\cdot price^{val}(A)$ value queries.
\end{proof}

All that is left is to show that $q_k$ from the statement of the claim can be bounded from above by $poly(mc(A),price^{val}(A),m)$. We now give a recursive algorithm for learning a valuation $v$ that is $k$-useless.

\vspace{0.1in} \noindent $FindUseless(S, Allowed)$
\begin{enumerate}
\item If $S$ is useless then return $S$.
\item If $v(S)=1$ then return $\emptyset$.
\item Initialize $UselessBundles=\emptyset$.
\item For each item $j\in Allowed$:
	\begin{enumerate}
	\item Remove $j$ from $Allowed$.
	\item Add $FindUseless(S+\{j\},Allowed)$ to $UselessBundles$.
	\end{enumerate}
\item Return $UselessBundles$.
\end{enumerate}
We will now see that running $FindUseless(\emptyset,M)$ returns the set of all useless bundles:
\begin{claim}
$FindUseless(S,Allowed)$ returns the set of useless bundles that contain $S$ and are contained in $S+Allowed$.
\end{claim}
\begin{proof}
We prove this by induction on the size of $Allowed$. If $Allowed=\emptyset$ then indeed $FindUseless(S,\emptyset)$ returns $S$ if and only if $S$ is useless.

Assume correctness for $|Allowed|=l$ and prove for $|Allowed|=l+1$. If $S$ is useless or $v(S)=1$ then the $FindUseless$ correctly terminates in the first two lines. To analyze the other case, suppose without loss of generality that $Allowed=\{1,2,3,\ldots, |Allowed|\}$. The main loop will first add to the set $UselessBundles$ all useless bundles that are contained in $S+Allowed$ and contain both $S$ and item $1$ (applying the induction hypothesis), then all useless bundles that contain item $2$ but not item $1$, then all useless bundles that contain item $3$ but not items $1,2$ and so on. The claim is completed since every useless bundle that contains $S$ and is contained in $S+Allowed$ must fall into one of these disjoint categories.
\end{proof}

\begin{claim}
The total number of value queries that $FindUseless(\emptyset,M)$ makes is $poly(k,m)$.
\end{claim}
\begin{proof}
Observe that if the call $FindUseless(S,Allowed)$ was executed ($S$ is ``visited''), there will be no other call to $FindUselss(S,Allowed')$, for any value of $Allowed'$. Furthermore, $FindUseless(S, Allowed)$ calls only to $FindUseless(S+j,Allowed')$. Thus, the set of bundles that are visited is a tree rooted at $\emptyset$. Notice that the leafs of the trees are either useless bundles or bundles that are valued $1$. Furthermore, each leaf $S$ with value $1$ was called from some node $S-\{j\}$ (for some $j$) with value $0$ and is on a path to a useless bundle.

Since there are at most $k$ useless bundles and since the length of the path from a root to a useless bundle is at most $m$, we get that the total number of bundles that are on a path from the root to a useless set is at most $m\cdot k$. Each such bundle has at most $m$ neighbors with value $1$ in the tree, so the total number of bundles that the algorithm visits is at most $m^2\cdot k$.

Notice that, ignoring recursive calls, each call to $FindUseless(S,Allowed)$ makes at most $m+1$ value queries (in the first two lines, at most $m$ queries to check whether $S$ is useless an additional one for $v(S)$). The total number of value queries that $FindUseless(\emptyset, M)$ therefore makes is at most $(m+1)\cdot m^2\cdot k$.
\end{proof}

All this gives us that $q_k\leq poly(k,m)$. 
%
This completes the proof of Theorem \ref{thm-value-reconstruct}.
\end{proof}

Define $tie^{val}(A)$ in an analogous way to $tie(A)$ (see definition in Section \ref{sec-menu-construction}). We conclude:
\begin{theorem}[characterization of the value query complexity of truthful mechanisms]
Fix a mechanism $A$ that makes only value queries and is truthful for general valuations. Then:
$$
\frac {mc(A)+price^{val}(A)+tie^{val}(A)} 3 \leq val(A) \leq poly(mc(A),price^{val}(A),tie^{val}(A),m,n) 
$$
\end{theorem}
The proof is very similar to the proof of Theorem \ref{thm-characterization} and is omitted.

%
%

\subsection{Demand Queries}\label{sec-demand}

In this section we consider mechanisms that access the valuations via demand and value queries only (but recall that a value query can be simulated by a polynomial number of value queries \cite{BN07}).

\subsubsection{An Impossibility Result for Menu Reconstruction}

Up until now, we have showed two menu reconstruction theorems. The first one (Theorem \ref{thm-construct-menu}) showed that we can find the menu using $poly(tax(A),price(A),m,n)$ communication. The second one (Theorem \ref{thm-value-reconstruct}) uses $poly(mc(A),price^{val}(A),m)$ value queries. In particular, for rich enough domain the running time of menu reconstruction is within a polynomial factor of the running time of the truthful mechanism.

We now show that no analogous result exists if the mechanism accesses the valuations using demand queries only. Specifically, we will show that if player $i$'s valuation is general, then there is a two-player mechanism $A$ with $price^{dem}(A)=1$ that makes $poly(m)$ demand queries but that reconstructing the menu presented to player $i$ takes $exp(m)$ demand queries.

\begin{theorem}\label{thm-demand-impossibility}
There is a two-player mechanism $A$ that is truthful for player $2$ with a general valuation such that $dem(A)=poly(m)$ but finding the menu presented to player $2$ requires $exp(m)$ demand queries to player $1$'s valuation.
\end{theorem}
\begin{proof}
Let $\mathcal M_T$ be the following menu, for every $T$ such that $|T|=\frac m 2$:
$$
\mathcal M_T(S)= \begin{cases} 
 |S| &  S\neq T \\
|S|+\frac 1 2 & S=T 
\end{cases}
$$
The mechanism is defined as follows. Player $1$ never gets any items. Player $2$ receives his profit-maximizing bundle from $\mathcal M_T$, where ties are arbitrarily broken. The identity of $T$ will depend on player $1$'s valuation, but we will have that $price^{dem}(A)=1$. We first show that by making $poly(m)$ demand queries we can find player $2$'s profit maximizing bundle, even if $T$ is unknown in advance.

\paragraph{An Algorithm for finding a profit maximizing bundle.} Start with a demand query with a price per item of $1$. Let $S_0$ be the profit maximizing bundle returned by that demand query. Use $price^{dem}(A)$ demand queries to check whether $S_0=T$, which happens if and only if $\mathcal M_T (S_0)=\frac m 2 +\frac 1 2$. If $S_0\neq T$, then we will show that $S_0$ is a profit maximizing bundle. 

If $S_0=T$, run additional $m$ demand queries. First, we run the following $\frac m 2$ demand queries, one for each item $j\in T$: in the $j$'th query the price of item $j$ is $\infty$ and the price of every other item is $1$. Let $R_1$ be the profit maximizing bundle among the results of all these $\frac m 2$ queries. The second batch of $\frac m 2$ demand queries, one for each $j\notin T$, sets a price $\frac 1 2$ for item $j$, a price $0$ for every item $j\in T$ and price of $1$ for all other items. Let $R_2$ be the profit maximizing bundle among the results of all queries in the second batch. We now show that one of $S,R_1,R_2$ is a profit maximizing bundle.

\begin{lemma}
Even if $T$ is unknown, the algorithm finds a profit-maximizing bundle by making $m+1+ price^{dem}(A)$ demand queries.
\end{lemma}
\begin{proof}
We first show that if $S\neq T$ then indeed $S$ is a profit maximizing bundle. To see that, observe that by the result of the first demand query we have that for all $S'$, $v_2(S)-|S|\geq v_2(S')-|S'|\geq v_2(T)-|T|>v_2(T)-|T|-\frac 1 2$. Thus, we get that $v_2(S)-\mathcal M_T(S)\geq v_2(S')-\mathcal M_T(S')$ for all $S'$, as needed.

We now show that if $S=T$ then one of $S,R_1, R_2$ is a profit maximizing bundle. We first claim that among all bundles that do not contain $T$, $R_1$ is a profit maximizing bundle. Denote by $Q$ the profit maximizing bundle among all bundles that do not contain $T$. Since $Q$ does not contain $T$, there is some item $j\in Q$ such that $j\notin T$. Consider the demand query that sets the price of $j$ to be $\infty$ (and the price of every other item to be $1$). Let $Q'$ be the bundle it returns. We have that for all $S'$ that do not contain $j$, $v_2(Q')-|Q'|\geq v_2(S')-|S'|$. Since we also have by assumption that $v_2(Q)-|Q|\geq v_2(Q')-|Q'|$. By the way we choose $R_1$, the profit of $Q$ is at most the profit of $R_1$.

We now show that $R_2$ is a profit maximizing bundle among all bundles that contain $T$. Denote by $Q$ some bundle with a profit maximizing bundle among all those that contain $T$. If the profit of $T$ is at least the profit of $Q$ then we are done, because the algorithm chooses a profit maximizing bundle among $T,R_1,R_2$. Thus, assume that $Q$ is strictly more profitable than $T$, hence $Q$ strictly contains $T$. Let $j\in Q$ be some item such that $j\notin T$.

Consider the $j$'th demand query in the second batch and denote by $Q'$ its answer. Notice that the bundle $Q'\cup T$ has at least the same profit as the bundle $Q'$, simply by the monotonicity of the valuations and since the price of every item in $T$ is $0$. Also notice that in this demand query the price of every bundle $S'$ that contains $T$ and $j$ is exactly $\mathcal M_T(S')-\mathcal M_T(T)$. Thus, the profit of $Q'$ in $\mathcal M_T$ is at least that of any other bundle that contains item $j$ and $T$. In particular, the profit is at least that of $Q$. This gives us that the profit of $R_2$ in $\mathcal M_T$ is at least that of $Q$, which finishes the proof of the lemma.
\end{proof}

To finish the proof, we describe how $T$ is determined. The idea is to embed into $v_1$ some problem that can only be solved by making exponentially many demand queries. The solution to this hard problem is $T$. This will give us that reconstructing the menu requires $exp(m)$ demand queries. The crux is that verifying whether a specific bundle $T$ is a solution to the hard problem can be done with only one query. Thus, $price^{dem}(A)$ can be computed with one demand query (to compute $\mathcal M_T(S)$, we check whether $S$ solves the hard problem) and we can find a profit maximizing bundle by making only $poly(m)$ demand queries. 

Our hard problem will be finding a value maximizing bundle of size $\frac m 2$. Fix some $\epsilon>0$. We will assume that there exists some bundle $T$ for which player $1$'s valuation is:
$$
v^T_1(S)= \begin{cases} 
 0 &  |S|<\frac m 2, \\
 0 & |S|=\frac m 2, S\neq T\\
\frac 1 4 & S=T, \\
1 & |S|>\frac m 2.
\end{cases}
$$
Notice that it is easy to check whether $S$ is a value-maximizing bundle: use one value query\footnote{Alternatively, use the demand query that sets $p_j=0$ for every $j\in S$ and $p_j=\infty$ otherwise.} and get $v_1(S)$. $S$ maximizes the value if and only if $v_1(S)=\frac 1 4$. This give us the following implementation of the price computing procedure for $A$: given bundle $S$, return $|S|+\frac 1 2$ if $v_1(S)=\frac 1 4$. Else, return $|S|$. All that is left to prove is that finding $T$ requires exponentially many demand queries to $v_1$. The proof is in fact a special case of a proof that was given in \cite{BDO12} and we bring the proof of the special case here for completeness.

\begin{lemma}
Finding $T$ requires $exp(m)$ demand queries.
\end{lemma}
\begin{proof}
We say that a demand query $p=(p_1,\ldots, p_m)$ \emph{covers} $T$, $|T|=\frac m 2$, if $T$ is the answer to the demand query $p$ when the valuation of player $1$ is $v^T_1$. Notice that if $T_p$ is the set of bundles that a demand query $p$ covers, we can replace every demand query by querying the value of each bundle $S\in T_p$. This is true since answer of the demand query can either by some bundle $T_p$, or some bundle of size different than $\frac m 2$. However, the value of every bundle with size different than $\frac m 2$ is fixed, a profit-maximizing bundle among those can be easily computed with no queries.

We show that finding a profit maximizing bundle using value queries only requires $exp(m)$ value queries. The proof will be concluded by showing that every demand query covers at most one bundle, which implies by our discussion above that every demand query can be simulated by one value query. Thus, for a constant $\epsilon>0$, $exp(m)$ demand queries are needed to find $T$, which finishes the proof of Theorem \ref{thm-demand-impossibility}.

\begin{claim}
Finding $T$ using value queries only requires ${m \choose {\frac m 2}}-1$ value queries.
\end{claim}
\begin{proof}
Fix some algorithm that finds $T$ and uses only value queries. For every one of the first ${m \choose {\frac m 2}}-2$ queries to bundles of size $\frac m 2$ return $0$. $T$ can be either one of the two bundles that were not queried so far, and an extra value query is needed to decide which one of them is $T$.
\end{proof}

\begin{claim}
Every demand query $p=(p_1,\ldots, p_m)$ covers at most one bundle.
\end{claim}
\begin{proof}
We first claim that if $T$ is covered by $p$, then for every $j\notin T$ it must hold that $p_j<\frac 3 4$. Otherwise, the bundle $T\cup \{j\}$ has a strictly higher profit: $v^T_1(T\cup \{j\})-\Sigma_{j'\in T\cup \{j\}}p_{j'}=1-\Sigma_{j'\in T}p_{j'}-p_j>\frac 1 4-\Sigma_{j'\in T}p_{j'}  = v^T_1(T)-\Sigma_{j'\in T}p_{j'}$. This implies that $T$ is not covered by $p$, a contradiction.

Next, observe that if $T$ is covered by $p$, then there is no item $j\in T$ with $p_j> \frac 1 4$, otherwise the profit of $T$ is negative. 

Together, this gives us that the only bundle that might be covered by $p$ is the bundle that contains all items with price at most $\frac 1 4$. 
\end{proof}

We thus get that finding $T$ requires exponentially many demand queries.
\end{proof}

This concludes the proof of Theorem \ref{thm-demand-impossibility}.
\end{proof}

\begin{remark}
Notice that in the proof the valuation of player $2$ could be general, but player $1$ that is presenting the menu could not have a general valuation. This is because we needed to have exactly one solution to the ``hard problem''. If the valuation of player $1$ was allowed to be general, we could have ended up with multiple bundles of size $\frac m 2$ with value $\frac 1 4$, and finding a profit maximizing bundle could no longer be done with polynomially many demand queries.

In fact, to allow player $1$ to have general valuation, we need to have some problem that needs exponentially many demand queries to solve, and every instance has at most one unique solution or, equivalently, instances with more than one unique solution can be solved with $poly(m)$ demand queries (the equivalent of unique-SAT \cite{VV86}). Currently, we do not know whether such a problem exists. We do note however that we can use cryptographic constructs to get a similar (somewhat weaker) result for general valuations. For example, given a one way permutation $\pi$, we could use a very similar construction to the one in the proof, except that we let $v_1(S)=\frac 1 4$ if and only if $\pi(S)=0$ (where we slightly abuse notation here by using the binary representation of $S$). Thus, there is only one bundle $S$ with $v_1(S)=\frac 1 4$, and constructing the menu is equivalent to inverting the one way permutation $\pi$, even if computing demand queries takes $O(1)$ time.
\end{remark}

\subsubsection{The Affinity of Mechanisms}

We now characterize the structure of the menu in mechanisms that use only demand queries.

\begin{definition}
A menu $\mathcal M$ is called \emph{$\alpha$-min-affine} if there exists a set of $\alpha$ price vectors $\{p^i\}_i$ with each $p^i_j\in \mathbb R^+\cup \{0\}\cup \{\infty\}$ and a set of numbers $\{r^i\}_i$ such that for every bundle $S$ with $\mathcal M(S)<\infty$ we have that $\mathcal M(S)=\min_i\{\Sigma_{j\in S}p^i_j+r^i\}$. 
\end{definition}

A menu $\mathcal M$ is \emph{$(\alpha,\beta)$-almost min-affine} if it is min-affine with complexity $\alpha$ except for $\beta$ many bundles (which may have arbitrary prices, including $\infty$). Notice that an $(\alpha,\beta)$-almost min-affine menu can obviously be described by $\alpha$ price vectors in addition to $\beta$ numbers. When $\beta=0$, we say that the affinity of $A$, denoted $aff(A)$, is $\alpha$.

\begin{theorem}\label{thm-demand}
Let $A$ be a mechanism that is truthful for general valuations and uses only demand and value queries. If $A$ makes at most $\alpha$ demand queries and $\beta$ value queries then every menu presented in $A$ is $(\alpha, \beta)$-almost min-affine.  
\end{theorem}
\begin{proof}
Fix some menu $\mathcal M$ that may be presented to player $i$. Let $(v_1^{\mathcal M},\ldots,  v_{i-1}^{\mathcal M},v_{i+1}^{\mathcal M}, \ldots ,v_n^{\mathcal M})$ be some valuation profile that presents the menu $\mathcal M$ to player $i$. Let $v_i$ be the valuations where for each bundle $S$ with $\mathcal M(S)<\infty$ we have that $v_i(S)=\mathcal M(S)$. Set the value of bundles $S$ with $\mathcal M(S)=\infty$ to $v(S)=(m+1)B$ (recall that $B$ is an upper bound on the highest finite price in the menu -- see Section \ref{sec-preliminaries}).

Consider the oracle calls that $A$ makes to the valuation $v_i$ in the instance $(v^{\mathcal M}_{-i},v_i)$: at most $\alpha$ demand queries and $\beta$ value queries. Let the price vector of the $k$'th demand query be $p^k$ and let $D^k$ be the bundle returned by this demand query. Let $r^k=v_i(D^k)-\Sigma_{j\in D^k}p^k_j$. Let the set of bundles queried by the value queries be $T=\{T^1,\ldots ,T^\beta\}$.

For each price vector $p^k$, obtain a price vector $p'^k$ by replacing every $p^k_j$ such that $p^k_j>B$ with $p^k_j=\infty$. In addition, if for some $k$ we have that $r^k>B$ then we set $p^k_j=\infty$ for all $j$. Let $\mathcal M'$ be the $(\alpha,\beta)$-almost min-affine menu defined by the price vectors $\{p'^k\}_k$ and the numbers $\{r^k\}_k$, except that we set the price of every bundle $T^r\in T$ to be $\mathcal M'(T^r)=\mathcal M(T^r)$.

The theorem is obtained by showing that $\mathcal M$ and $\mathcal M'$ are equal:
\begin{lemma}\label{lemma-demand-queries-at-most}
For every $k$ and every bundle $S\notin T$ with $\mathcal M(S)<\infty$ we have that $\mathcal M(S)=v_i(S)\leq \Sigma_{j\in S}p^k_j+r^k$.
\end{lemma}
\begin{proof}
Suppose towards a contradiction that for some $k$, $\mathcal M(S)=v_i(S)> \Sigma_{j\in S}p^k_j+r^k$. We then have that:
\begin{equation*}
v_i(S)-\Sigma_{j\in S}p^k_j > r^k=v_i(D^k)-\Sigma_{j\in D^k}p^k_j
\end{equation*}
In contradiction to our assumption that $D^k$ is a profit-maximizing bundle when the price vector is $p^k$.
\end{proof}

\begin{lemma}\label{lemma-demand-queries-exactly}
For every bundle $S\notin T$ with $\mathcal M(S)<\infty$ there exists some $k$ such that $\mathcal M(S)=v_i(S)= \Sigma_{j\in S}p^k_j+r^k$.
\end{lemma}
\begin{proof}
We say that bundle $S$ is \emph{tight} if $\mathcal M(S)<\infty$ and for some $L\neq \emptyset$, $v_i(S+L)=v_i(S)$. We start with proving the lemma for bundles that are not tight. Fix a bundle $S$ that is not tight, and let $v'_{i,\epsilon}$ be the valuation that is identical to $v$ except that $v_{i,\epsilon}(S)=v_i(S)+\epsilon$, for some small enough $\epsilon>0$ (so that $v_{i,\epsilon}$ is still monotone). Since for every bundle $S'$ with $\mathcal M(S')< \infty$ we have that $v_i(S')=\mathcal M(S')$, when the valuation of player $i$ is $v_i$, his profit is $0$, regardless of the bundle that is eventually allocated to him. Now observe that when the valuation of player $i$ is $v_{i,\epsilon}$, $S$ is the only bundle with a strictly positive profit is $S$. Thus, the mechanism must allocate $S$ to player $i$ when his valuation is $v_{i,\epsilon}$.

If the mechanism does not allocate $S$ to player $i$ when the valuation is $v_i$, the mechanism must distinguish between the case when the valuation is $v_i$ and the case when the valuation is $v_{i,\epsilon}$ (for every small enough $\epsilon>0$). Since the mechanism makes only two types of queries, this is possible only if when the valuation is $v_{i,\epsilon}$ we have that $S\in T$ or if there exists some $k$ such that the $k$'th demand query returns a bundle that is different than $D^k$. Since by assumption $S\notin T$ we assume the latter case. For every small enough $\epsilon>0$, let $k_\epsilon$ be the first demand query that changes (comparing to the sequence with the valuation $v_i$). Since the number of demand queries $\alpha$ is finite, when we take $\epsilon$ to $0$ there must be some $k$ such that for every $\delta>0$ there exists some $\epsilon<\delta$ such that the $k$'th demand query is the first to change when the valuation of player $i$ is $v_{i,\epsilon}$. I.e., for arbitrarily small $\epsilon>0$ we have that $v_i(S)+\epsilon> \Sigma_{j\in S}p^k_j+r^k$, and therefore we also have that $v_i(S)\geq  \Sigma_{j\in S}p^k_j+r^k$. Applying Lemma \ref{lemma-demand-queries-at-most} we get that $v_i(S)=  \Sigma_{j\in S}p^k_j+r^k$, as needed.

We are left with proving the lemma for tight bundles. Fix a tight bundle $S$, and let $L$, $L\cap S=\emptyset$, be a maximal set for which $v_i(S+L)=v_i(S)$. If $S\in \{D^k\}_k$, then it immediately holds that for some $k$, $v_i(S)=\Sigma_{j\in S}p^k_{j}+r^k$. If $S\notin \{D^k\}_k$, we observe that $S+L$ is not tight (if it is tight then there exists $L\subseteq L'$ such that $v_i(S)=v_i(S+L')$, contradicting the maximality of $L$) and thus there is some $k$ for which $v_i(S+L)=\Sigma_{j\in S+L}p^k_{j}+r^k$. Since $v_i(S+L)=v_i(L)$, it must be that $\Sigma_{j\in L}p^k_j=0$, as otherwise $v_i(S)-\Sigma_{j\in S}p_j>r_k$. Thus $S$ has to be the result of this query. In particular, $v_i(S)=v_i(S+L)=\Sigma_{j\in S+L}p^k_{j}+r^k=\Sigma_{j\in S}p^k_{j}+r^k$, which completes the proof.
\end{proof}

\begin{lemma}\label{lemma-demand-infty}
For every bundle $S\notin T$ we have that $\mathcal M(S)=\infty$ if and only if for every $k$, $\Sigma_{j\in S}p^k_j+r^k> (m+1)\cdot B$.
\end{lemma}
\begin{proof}
Consider a bundle $S\notin T$ with $\mathcal M(S)=\infty$. If $S$ is not profitable in the $k$'th demand query, then by definition $\Sigma_{j\in S}p^k_j\geq v(S)$ and the lemma follows. If $S$ is profitable we know that:
$$
r^k=v_i(T^k)-\Sigma_{j\in T^k}p^k_j\geq v_i(S)-\Sigma_{j\in S}p^k_j=(m+1)B-\Sigma_{j\in S}p^k_j
$$
Rearranging, we get that in this case $\Sigma_{j\in S}p^k_j+r^k> (m+1)B$.
\end{proof}

Thus, the proof of Theorem \ref{thm-demand} can be concluded as follows: if $S\in T$, we trivially have that $\mathcal M(S)=\mathcal M'(S)$. For every bundle $S\notin T$ with $\mathcal M(S)<\infty$, Lemma \ref{lemma-demand-queries-exactly} gives us that there exists some $k$ for which $v_i(S)=\mathcal M(S)= \Sigma_{j\in S}p^k_j+r^k$. In particular, since $\mathcal M(S)\leq B$, for this $k$ it must hold for every $j\in S$ that $p^k_j\leq B$. Thus, we have that $\Sigma_{j\in S}p^k_j=\Sigma_{j\in S}p'^k_j$ and also that $r^k\leq B$. By Lemma \ref{lemma-demand-queries-at-most} the price of the bundle cannot be higher, which gives us that $\mathcal M(S)=\mathcal M'(S)$ for every $S\notin T$.

Now for bundles $S\notin T$ with $\mathcal M(S)=\infty$. By Lemma \ref{lemma-demand-infty}, for every $k$ it holds that $\Sigma_{j\in S}p^k_j+r^k> (m+1)\cdot B$. Since $|S|\leq m$, the LHS consists of at most $m+1$ non-negative summands, thus one of them is greater than $B$. We will show that for at least one $j\in S$, $p'^k_j=\infty$. Hence, $\mathcal M'(S)=\infty$, as needed. If we have some $j\in S$ such that $p^k_j>B$, then $p'^k_j=\infty$. The only other option is that $r^k>B$. In this case, for all $j$ it is true that $p'^k_j=\infty$, which concludes the proof.
\end{proof}

The characterization gives us some hope of finally proving some bounds on the power of computationally efficient mechanisms that use only demand queries to access the valuations. When restricting ourselves to deterministic mechanisms, a ratio of $O(\sqrt m)$ is the best known \cite{DNS05}. When randomization is allowed, a significantly better approximation ratio of $O(\sqrt {\log m})$ is possible \cite{D16}.
\begin{open}
Consider mechanisms for combinatorial auctions with $m$ items that make only demand queries and are truthful for submodular valuations.
\begin{itemize}
\item Is there a \emph{deterministic} mechanism with $aff(A)=poly(m,n)$ that obtains an approximation ratio of $m^{\frac 1 2-\epsilon}$, for some constant $\epsilon>0$? 

\item Is there a \emph{randomized universally truthful} mechanism with $aff(A)=poly(m,n)$ that obtains an approximation ratio of $({\log m})^{\frac 1 2-\epsilon}$, for some constant $\epsilon>0$? 
\end{itemize}
\end{open}

Interestingly, the $O(\sqrt {\log m})$ approximation ratio of \cite{D16} is obtained by a mechanism with $aff(A)=1$. Thus, even the following easier question is of interest:
\begin{open}
Let $A$ be a randomized universally truthful mechanism for combinatorial auctions with $m$ that is truthful for submodular valuations. Suppose that $aff(A)=1$. Can the approximation ratio of $A$ be $({\log m})^{\frac 1 2-\epsilon}$, for some constant $\epsilon>0$? 
\end{open}

\paragraph{Tightness.} To see that the Theorem \ref{thm-demand} is essentially tight (i.e., there is a mechanism with an $(\alpha,\beta)$-min affine menu that makes approximately $\alpha$ demand queries and $\beta$ value queries), consider the following truthful mechanism for combinatorial auctions with two players, Alice and Bob. We first show this for some $\alpha>0$ and $\beta=1$, and will sketch how to generalize for any $\beta$ later. Fix some set of $M'$ items, $|M'|=\frac m 2$. 

Let $M^{Bob}=\{\mathcal M_1, \ldots, \mathcal M_{2^c}\}$ be a set of $2^c$ menus, where each $\mathcal M_i\in M^{Bob}$ is $(\alpha,0)$-min affine. We require in addition that for every $t$ and item $j\notin M'$, $\mathcal M_t(T)=\infty$, for every $T$ that contains item $j$. Bob then chooses a profit maximizing bundle according to the menu $\mathcal M_t$ (breaking ties in some consistent way), and pays appropriately. Alice never receives any items. We claim that it is possible to find a profit-maximizing bundle by making $\alpha$ demand queries:

\begin{claim}\label{claim-demand-optimize}
Let $\mathcal M$ be some menu $(\alpha,0)$ that is presented to player $i$ with valuation $v_i$. Then, a profit maximizing bundle according to $\mathcal M$ can be found by making $\alpha$ demand queries.
\end{claim}
\begin{proof}
Let $\{p^i\}_i$ be the price vectors and $\{r^i\}_i$ be the numbers that define the min affine menu. We will find a profit maximizing bundle by making $\alpha$ demand queries, one for each $p^i$. Let $D^k$ be the bundle returned by the $k$'th demand query. We claim that the maximum profit is obtained by a bundle $D^k\in \arg\max_{k'} v_i(D^{k'})-(\Sigma_{j\in D^{k'}}p^{k'}_j+r^{k'})$. If this maximum profit is non-negative, then the profit maximizing bundle is the empty set.

In order to prove this, suppose that some other bundle $T$ maximizes the profit. Let $k'\in \arg\min_k(\Sigma_{j\in T}p^k_j+r^k)$. But then, $v_i(D^{k'})-\mathcal M(D^k)= v_i(D^{k'})-\Sigma_{j\in D^{k'}}p^{k'}_j-r^{k'} \geq v_i(T)-\Sigma_{j\in T}p^{k'}_j-r^{k'}=v_i(T)-\mathcal M(T)$ (where in the inequality we use $v_i(D^{k'})-\Sigma_{j\in D^{k'}}p^{k'}_j \geq v_i(T)-\Sigma_{j\in T}p^{k'}_j$, since $D^k$ maximizes the profit in the $k'$ demand query), i.e., $D^{k'}$ is at least as profitable as $T$.
\end{proof}

The mechanism is clearly truthful. It uses one value query and $\alpha$ demand queries (by the claim). All menus that are presented are $(\alpha,1)$-min affine.

Finally, to extend this result to min-affine menus with $\beta>1$, embed a construction similar to the tightness example of Subsection \ref{sec-value} using only items that are not in $M'$.

\subsection{Non-Interactive Menu Optimization}

Here we study the complexity of non-interactive menu optimization. Consider the following two player menu optimization problem: Alice's input is some menu $\mathcal M\in U$, where the set of menus $U$ is known in advance. Bob's input is some valuation $v$. The goal is to find a bundle that maximizes the profit $v(S)-\mathcal M(S)$. We restrict ourselves to one way protocols: Alice speaks first and then Bob. After Bob speaks, both parties know a profit maximizing bundle $S$.

We consider this problem in three different models that differ on how Bob's valuation is accessed. In the general communication model Bob's message is not restricted in any way, as long as his message depends only on Alice's message and $v$. We then consider settings in which Bob's valuation can be accessed by one specific type of queries, either value or demand. That is, after Alice speaks the center makes a (possibly adaptive) sequence of queries to $v$ that is determined only by Alice's message. After the sequence of queries ends, the players know Bob's profit-maximizing bundle.

Let $U$ be the set of menus presented to some player in a truthful mechanism $A$. As we will see, the communication complexity of this problem in the general model is essentially $tax(A)$, in the value queries model $val(A)$, and $aff(A)$ in the demand queries model. That is, these measures capture the complexity of non-interactive menu optimization.

\begin{lemma}
In the general model, the communication complexity of the menu optimization problem is between $tax(A)$ and $tax(A)+m$.
\end{lemma}
\begin{proof}
To see that the communication complexity is at most $tax(A)+m$, consider the following protocol: Alice uses $tax(A)$ bits to send the index of the menu that she holds, Bob then uses $m$ bits to announce a profit maximizing bundle.

Suppose that there is a protocol with communication complexity strictly less than $tax(A)$. Then, there are two menus $\mathcal M, \mathcal M'$, $\mathcal M\neq \mathcal M'$ for which Alice sends the same message. Since the two menus are different, there is some bundle $S$ such that $\mathcal M(S)\neq \mathcal M'(S)$. Without loss of generality assume that $\mathcal M(S)\neq \mathcal M'(S)$. Suppose that Bob's valuation $v$ is a single minded valuation: $v(S)=\frac {\mathcal M(S)+ \mathcal M'(S)} 2$, every bundle $T$ that contains $S$ equals $v(S)$, and every other bundle equals $0$. Since Bob cannot distinguish between $\mathcal M$ and $\mathcal M'$, the bundle $T$ returned by the protocol is identical given that Bob's valuation is $v$. However, if the menu is $\mathcal M$ then any profit maximizing bundle does not contain $S$, but if the menu is $\mathcal M'$ every profit maximizing bundle must contain $S$. A contradiction.
\end{proof}

\begin{lemma}
In the value queries model, the communication complexity of the menu optimization problem is between $mc(A)\cdot k$ and $mc(A)\cdot m\cdot k+mc(A)\cdot k$, where $k$ is the number of bits that are used to represent prices in $A$.
\end{lemma}
\begin{proof}
To see that the communication complexity is at most $mc(A)\cdot k$, consider the following protocol: for each of the $mc(A)$ bundles that in $\mathcal M$, Alice uses $m$ bits to send the identity of each bundle that is in $\mathcal M$ and additional $k$ bits to send its price. Then, the center makes a value query to determine the value of every bundle that appeared in Alice's message.

The fact that the communication complexity is at least $mc(A)\cdot k$ follows from Theorem \ref{thm-value} that essentially shows that the center needs to query every bundle in the menu in order to find a profit maximizing bundle.
\end{proof}

\begin{lemma}
In the demand queries model, the communication complexity of the menu optimization problem is between $aff(A)\cdot k$ and $aff(A)\cdot (m+1)\cdot k$, where $k$ is the number of bits that are used to represent prices in $A$.
\end{lemma}
\begin{proof}
To see that the communication complexity is at most $aff(A)\cdot k$, consider the following protocol: for each of  the $aff(A)$ price vectors that define the min-affine menu, Alice sends the $m+1$ numbers that describe it. The center then makes the appropriate demand queries to find a profit-maximizing bundle as in Claim \ref{claim-demand-optimize}.

If the center can make $r$ demand queries on every menu to find a profit-maximizing bundle of $v$, then by Theorem \ref{thm-demand} the affinity of every menu is at most $r$, and thus $aff(A)\leq r$, as needed.
\end{proof}

\section{Applications and Extensions}

\subsection{From Ex-Post Nash to Dominant Strategy}\label{subsec-to-dominant}

The revelation principle implies that if there is a mechanism that implements some social choice function in an ex-post Nash equilibrium, there is also a mechanism that implements the same social choice function in a dominant strategy equilibrium. Unfortunately, the communication complexity of the latter mechanism might be exponential comparing to the communication complexity of the former (e.g., the already mentioned example of the VCG mechanism for gross substitutes). We provide a more efficient transformation.

\begin{proposition}\label{prop-dominant}
Let $A$ be a two player mechanism for combinatorial auctions that reaches an ex-post Nash equilibrium. Then, there is a mechanism $A'$ that implements the same social choice function in dominant strategies with communication complexity $2(tax(A)+m)+tie(A)\leq 2(tax(A)+m)+cc(A)$. In particular, if $A$ is truthful for general valuations, then the communication complexity of the new implementation is $3cc(A)+2m$.
\end{proposition}
In particular, for general valuations we pay ``almost nothing'' (communication-wise) for strengthening the solution concept (simply using the fact that for general valuations by Theorem \ref{thm-indexing-complexity} we have that $ tax(A)\leq cc(A)$). Similar transformations are possible of course for other classes of valuations using Theorem \ref{thm-indexing-complexity}. Before formally proving for Proposition \ref{prop-dominant} we provide some intuition. A naive proof for this proposition would be the following protocol:
\begin{enumerate}
\item Each player $i$ simultaneously sends $ tax(A)$ bits that denote the index of $\mathcal M_i$ -- the menu he presents to the other player.
\item Each player $i$ sends a description of some maximum profit bundle $T_i$ in the menu presented to him ($m$ bits for each player). Denote the price of $T_i$ in the menu by $p_i$.
\item Each player $i$ is assigned $T_i$ and pays $p_i$.
\end{enumerate}
This protocol ``almost works'' except that it is not clear how each player $i$ chooses which maximum-profit bundle to report if there are several bundles that maximize the profit. To solve this we have to be able to break ties correctly, and make sure that each player has a dominant strategy.

\begin{proof}(of Proposition \ref{prop-dominant})
We start with some definitions. Let $P$ be some protocol. Given strategies strategies $s_1(\cdot),\ldots,s_n(\cdot)$, we say that a (possibly partial) transcript $T$ of $P$ is \emph{consistent} with a valuation profile $(v_1,\ldots, v_n)$ if the transcript is $T$ when each player $i$ is playing $s_i(v_i)$. Consider a transcript $T$ of $P$ that is not consistent with any valuation profile $(v_1,\ldots, v_n)$. Let $T'$ be the minimal prefix of $T$ that is not consistent and let $i$ be the player that sent the last bit in $T'$. Player $i$ is the \emph{inconsistent player} of $T'$. 

For each player $i$, let $s_i$ denote the equilibrium strategy of each player $i$ in $A$. The mechanism $A'$ is the following:
\begin{enumerate}
\item Each player $i$ simultaneously sends $ tax(A)$ bits that denote the index of $\mathcal M_i$ -- the menu he presents to the other player in $s_i(v_i)$.
\item Each player $i$ sends a description of some maximum profit bundle $T_i$ according to the menu presented to him by the other player ($m$ bits for each player). Denote the price of $T_i$ in the menu by $p_i$.
\item\label{step-run-old} Run the mechanism $A$.
\item For each player $i$ let $s'_i(v_i)$ be the strategy where in the first step player $i$ sends the index of $\mathcal M_i$ and in Step \ref{step-run-old} player $i$ plays as in $s_i(v_i)$.
\item If there exist valuations $v_1,v_2$ such that the transcript is consistent with $s'_1(v_1)$ and $s'_2(v_2)$ then the outcome of $A'$ is identical to that of $A$. Otherwise, let $i$ be the inconsistent player. In this case, player $i$ is not allocated any bundle and pays nothing. The other player wins the bundle $T_i$ and pays $p_i$.
\end{enumerate}
Observe that if each player $i$ with valuation $v_i$ plays $s'_i(v_i)$ then the outcome of $A'$ is identical to the outcome of $A$ when each player $i$ plays $s_i(v_i)$. The statement of the proposition regarding the communication complexity of $A'$ is obvious as well. It remains to show that $s'_i$ is a dominant strategy. We start with two helper claims.

\begin{claim}\label{claim-inconsistent}
If player $i$'s strategy is $s'_i(v_i)$, for some $v_i$, then player $i$ is not an inconsistent player.
\end{claim}
\begin{proof}
Assume that $i=1$, but the proof is essentially the same for $i=2$. Let $q_2$ be the strategy of player $2$. Consider a run of $A'$. If player $2$ is an inconsistent player then there is nothing to prove. Therefore, we will consider the messages sent by player $1$ one by one, in each point assuming that the transcript so far is consistent with some strategies $(s'_1(v_1), s'_2(v_2))$. Now consider player $1$ sending his next message according to $s'_1(v_1)$. Notice that this message is identical to the message that is sent at this point in the transcript where the players use strategies $s'_1(v_1)$ and $s'_2(v_2)$. In particular the next message according to $s'_1(v_1)$ that player $1$ sends does not make him inconsistent since the partial transcript is identical to the prefix of the final transcript when both players are playing according to $s'_1(v_1)$ and $s'_2(v_2)$.
\end{proof}

\begin{claim}\label{claim-max-profit}
If player $i$ with valuation $v_i$ uses the strategy $s'_i(v_i)$ then his profit is $v_i(T_i)-p_i$.
\end{claim}
\begin{proof}
As before, assume that $i=1$, the proof is essentially the same for $i=2$. First, by Claim \ref{claim-inconsistent} player $1$ is not an inconsistent player. If player $2$ is the inconsistent player, then by the definition of the protocol player $1$ is assigned $T_1$ and pays $p_1$ so his profit his $v_1(T_1)-p_1$, as needed.

We therefore assume that the strategies that the players play are consistent with some strategies $(s'_1(v_1),s'_2(v_2))$. Denote the bundle that player $1$ got by $T'_1$ and his payment by $p'_1$. Now recall that the outcome of $A'$ with these strategies is identical to the outcome of $A$ with the strategies $(s_1(v_1),s_2(v_2))$ and that since these strategies form an ex-post Nash equilibrium in $A$, it must be that $T'_1$ is a maximum-profit bundle according to the menu $\mathcal M_2$. However, $T_1$ is also a maximum profit bundle according to $\mathcal M_2$ and thus $v_1(T'_1)-p'_1=v_1(T_1)-p_1$, which finishes the proof.
\end{proof}

To conclude the proof it suffices to show that:
\begin{claim}
For every player $i$ with valuation $v_i$, $s'_i(v_i)$ is a dominant strategy in $A'$.
\end{claim} 
\begin{proof}
We prove the claim for $i=1$, the claim for $i=2$ is essentially the same. Fix a strategy $q_i$ for every player $i$. If player $2$ is inconsistent then by claim \ref{claim-max-profit} the profit of player $1$ is $v_1(T_1)-p_1$. If player $1$ plays a different strategy that makes player $2$ consistent, the menu $\mathcal M_2$ must still be consistent with the menu presented to him in the first step, otherwise player $2$ is inconsistent. Since $T_1$ maximizes the profit, the profit of player $1$ is at most $v_1(T_1)-p_1$.

Let $S_1$ be the bundle that player $1$ is allocated at the end of $A'$. Let $\mathcal M_2$ be the menu that player $1$ is presented with at the first step of the protocol. We will show that the payment of player $1$ is $\mathcal M_2(S_1)$ -- the price of $S_1$ in $\mathcal M_2$. Claim \ref{claim-max-profit} gives us that player $1$ is not worse off playing $s'_1(v_1)$ in this case. If player $1$ is inconsistent then his profit is $0$ and, again, he is not worse off playing $s'_1(v_1)$, which completes the proof.

We now show that if player $1$ is consistent then the payment of player $1$ is $\mathcal M_2(S_1)$. This is obvious if player $2$ is inconsistent. Else, the strategies are consistent with some $(s'_1(v'_1),s'_2(v'_2))$. Since player $2$ is consistent, we in particular have that when player $2$ uses the strategy $s_2(v'_2)$ he presents to player $2$ the menu $\mathcal M_2$ and player $1$ chooses a profit-maximizing bundle from that menu. Since the outcome of $A$ with strategies $(s_1(v'_1),s_2(v'_2))$ is identical to the outcome of $A'$ with strategies $(s'_1(v'_1),s'_2(v'_2))$ we have that the price of $S_1$ in $A'$ is determined according to $\mathcal M_2$.
\end{proof}

This concludes the proof of the proposition.
\end{proof}

\begin{open}
Is there a social choice function $f$ (for three players or more) for combinatorial auctions with general valuations such that:
\begin{itemize}
\item There exists a protocol with communication complexity $poly(m)$ that implements $f$ in ex-post Nash equilibrium.
\item Any dominant-strategy implementation of $f$ requires $exp(m)$ bits.
\end{itemize}
\end{open}

\subsection{Taxation, Welfare Maximization, and Simultaneous Algorithms}\label{subsec-welfare}

We now explore the connections between truthful mechanisms and simultaneous algorithms \cite{DNO14}. We first show that if there is a two-player truthful mechanism that provides an approximation ratio of $\alpha$, then there is an $\alpha$-approximation \emph{simultaneous} algorithm with essentially the same communication complexity. Thus, in order to prove a separation between the power of computationally efficient truthful mechanisms for welfare maximization and their non-truthful counterparts it is enough to bound the power of simultaneous algorithms for two players.

We start with solving a special case. Let $A$ be a mechanism that is truthful for a class $\mathcal V$ of valuations. We say that $A$ is \emph{precise} if for every player $i$, every menu $\mathcal M$ that may be presented to player $i$ and every $v\in \mathcal V$, there is always only a single bundle that maximizes the profit of $i$, i.e., for every $v\in \mathcal V$: $|\arg\max_Sv(S)-\mathcal M(S)|=1$.

\begin{claim}\label{claim-precise-mechanism}
Let $A$ be a precise two-player mechanism that is truthful for some class of valuations $\mathcal V$. Denote its approximation ratio to the welfare by $\alpha$. Then, there is a simultaneous algorithm $A'$ that provides an $\alpha$-approximation for $\mathcal V$ with communication complexity $2\cdot tax(A)$.
\end{claim}
\begin{proof}
For every menu $\mathcal M$ presented to player $1$ by player $2$ and menu $\mathcal M'$ presented by player $1$ to player $2$, let $S^{\mathcal M,\mathcal M'}$ be the union of the bundles that player $1$ might receive when the menus presented are $\mathcal M$ and $\mathcal M'$. I.e., if we let $\mathcal V^{\mathcal M'}\subseteq \mathcal V$ be the subset of valuations in $\mathcal V$ such that player $i$ presents the menu $\mathcal M'$: 
$$S^{\mathcal M,\mathcal M'}=\cup_{v\in \mathcal V^{\mathcal M'}}\arg\max_Sv(S)-\mathcal M(S)$$
The simultaneous algorithm $A'$ can now be defined as follows. Each player $i$ simultaneously sends $tax(A)$ bits that represent the index of the menu that player $i$ is presenting to the other player in $A$. Let $\mathcal M$ and $\mathcal M'$ be those menus. The output of the algorithm $A'$ is $(S^{\mathcal M,\mathcal M'}, M-S^{\mathcal M,\mathcal M'})$.

Consider some instance $(v,v')$. Denote the allocation of $A$ in that instance by $(S_1,S_2)$ and that of $A'$ by $(S'_1,S'_2)$. We will show that $S_1\subseteq S'_1$ and that $S_2\subseteq S'_2$ and the claim regarding the approximation ratio immediately follows by using the monotonicity of the valuations.

We first observe that $S_1\subseteq S'_1$. This is true since $\arg\max_S v(S)-\mathcal M(S)\subseteq S^{\mathcal M,\mathcal M'}=S'_1$ by definition. To show that $S_2\subseteq S'_2$ we will show that $S_2\cap S^{\mathcal M,\mathcal M'}=\emptyset$. Else, there exists some $\overline v\in \mathcal V^{\mathcal M'}$ such that for the (unique) profit maximizing bundle $S$ in $\mathcal M$ we have $S\cap S_2\neq \emptyset$. Now observe in the instance $(\overline v,v')$ $A$ allocates the bundle $S$ to $\overline v$ and the bundle $S_2$ to $v'$, since by preciseness they are both the uniquely profit-maximizing bundles. However, this is not a valid allocation since $S\cap S_2\neq \emptyset$.
\end{proof}

Not every mechanism is precise since there might be several bundles that maximize the profit. However, since the profit from a bundle is defined by a linear inequality, adding independent ``random noise'' to the value of each bundle results with high probability with a unique bundle that maximizes the profit in every menu. Therefore, we will show that for every valuation $v$ there is a valuation $v'$ that is ``almost the same'' such that if we look at $A$ with respect to all possible valuations $v'$ we get that $A$ is precise and the result follows. One technicality that arises is that we need to make sure that the precision of representing numbers in $A$ is big enough so that the noise we add is small enough and the profit maximizing bundle is unique.

\begin{definition}
Valuation $v$ \emph{$\epsilon$-approximates} a valuation $v'$ if for every bundle $S$ it holds that $|v(S)-v'(S)|\leq \epsilon$. A class of valuations $\mathcal V$ \emph{$\epsilon$-approximates} a class of valuations $\mathcal V'\subseteq \mathcal V$ if for every $v\in \mathcal V$ there is some $v'\in \mathcal V'$ that $\epsilon$-approximates it. 
\end{definition}

We prove the next proposition for the case that $\mathcal V$ is the class of general (monotone) valuations. The extension to the other classes we discuss in this paper should be clear. 

\begin{proposition}\label{prop-close-valuations}
Fix some $\epsilon>0$. Let $A$ be a two-player mechanism that is truthful for general valuations. Suppose that numbers are represented in $A$ (in the sense of Section \ref{subsec-representation}) so that in an interval of $\frac \epsilon 2$ there are $l$ distinct numbers, $l>2^{2m+ tax(A)}$. Then, there is a class of valuations $\mathcal V'$ that $\epsilon$-approximates the class of general valuations such that $A$ is precise with respect to $\mathcal V'$. 
\end{proposition}
\begin{proof}
Given a valuation $v$, we first make sure that for every $S\subseteq S'$ we have that $v(S')\geq v(S)+\frac \epsilon {2m}$. This can easily be done by increasing the value of every bundle $S$ by $\frac {|S|} {2m} \cdot \epsilon$.

Our next step is to construct at random a new valuation $v'$ by setting for each bundle $S$, $v'(S)=v(S)+r_S$, where $r_S$ is one of the $l$ numbers between $0$ and $\frac \epsilon {2m}$, selected uniformly at random and independently for each $r_S$. Notice that $v'$ is still a monotone valuation and that it $ \epsilon $-approximates $v$.

We now compute the probability that $v'$ has a unique profit-maximizing bundle in some menu $\mathcal M$. Consider two bundles $S$ and $S'$. The probability that $v(S)-\mathcal M(S)+r_S=v(S')-\mathcal M(S')+r_{S'}$ is at most $\frac 1 {l}$ because of the size of the set that the $r_S$'s are selected from. Using the union bound and going over all $2^{2m}$ pairs, we get that with probability at most $\frac {2^{2m}} {l}$ no two bundles have the same profit in $\mathcal M$. Using the union bound again, the probability that there is a unique bundle that maximizes the profit in every menu that may be presented is at most $\frac {2^{2m+tax(A)}} {l}<1$. Thus there is a $v'$ that $\epsilon$-approximates $v$ and furthermore, for every menu that is presented to $v'$ there is a unique bundle that maximizes the profit. Finally, the class $\mathcal V'$ can now be defined as the set of all $v'$ generated from $v\in \mathcal V$ as above.
\end{proof}

\begin{corollary}\label{corollary-simultaneous}
Let $A$ be a two-player mechanism that is truthful for general valuations. Suppose that numbers are represented in $A$ (in the sense of Section \ref{subsec-representation}) so that in an interval of $\frac \epsilon 2$ there are at least $l=2^{2m+ tax(A)}$ distinct numbers. Denote the approximation ratio to the welfare by $\alpha$. Then, there is a simultaneous algorithm $A'$ with communication complexity $2tax(A)$, that outputs for every instance an allocation with welfare $ {\frac {OPT} \alpha-\epsilon}$, where OPT is the value of a welfare maximizing allocation in the instance.
\end{corollary}
Before proving the corollary, two comments on its limitations are in place. We first discuss the required precision $l$. The number of bits needed to represent $l$ numbers is $\log l$, and thus we require that numbers will be represented by $poly(m,tax(A))$ bits. In particular, for the domains handled in Theorem \ref{thm-indexing-complexity} and for mechanisms with polynomial communication (which are the domains and mechanisms for which this approach of proving impossibility is of interest), we get that we need $poly(m)$ bits to represent numbers, as common and expected.

Second, the approximation guarantee suffers from an additive loss of $\epsilon$. This almost does not affect the approximation ratio if the valuations are large enough. Moreover, essentially all reasonable valuation classes are ``scalable'' in the sense that any valuation is still in the class when multiplying all values by the same constant. Therefore, if the original valuation of a player is ``too small'' he can simply play as if his valuation is his real valuation times a large enough scale factor, in addition to sending the scale factor. If one of the valuations does not need scaling, then the smaller valuation can be ignored for the purpose of analyzing the approximation ratio. If both valuations need scaling to essentially the same factor, the algorithm mimics the allocation of the scaled up valuations, and the (multiplicative) approximation ratio is preserved with respect to the actual valuations. 

\begin{proof}(of Corollary \ref{corollary-simultaneous})
The simultaneous algorithm is the following: each player $i$ maps his valuation to some valuation that $\epsilon$-approximates it in the set $\mathcal V'$ guaranteed by Proposition \ref{prop-close-valuations}. Then the players run the algorithm of Claim \ref{claim-precise-mechanism} with respect to their mapped valuations. The approximation ratio follows since the mapped valuations are an $\epsilon$-approximation to the actual valuations of the players.
\end{proof}

Thus, in order to prove the first gap between computationally efficient truthful mechanisms for combinatorial auctions and their non-truthful counterparts, it is enough to settle on the affirmative one of the following questions (the best $2$-player algorithm for XOS valuations obtains an approximation ratio of $\frac 4 3$ \cite{DS06} and for submodular valuations the best currently known algorithm obtains an approximation ratio of $\frac {17} {13}$ \cite{FV06}):
\begin{open}
\ 
\begin{itemize}
\item Is there a $2$-player simultaneous algorithm for combinatorial auctions with XOS valuations that provides a $\frac 4 3$-approximation with message length $poly(m)$? 
\item Is there a $2$-player simultaneous algorithm for combinatorial auctions with submodular valuations that provides a $\frac {17} {13}$-approximation with message length $poly(m)$? 
\end{itemize}
\end{open}

As discussed in the introduction, we can alternatively prove impossibility results by proving lower bounds on the taxation complexity, e.g., for combinatorial auctions with $n$ subadditive players:
\begin{open}
\ 
\begin{itemize}
\item Let $A$ be a truthful deterministic mechanism for combinatorial auctions with $n$ subadditive players that obtains an approximation ratio of $m^{\frac 1 2 -\epsilon}$, for some constant $\epsilon>0$. Is the taxation complexity of $A$ exponential in $n$ and $m$?
\item Let $A$ be a randomized universally truthful mechanism for combinatorial auctions with $n$ XOS players that obtains an approximation ratio of $(\log m)^{\frac 1 2 -\epsilon}$, for some constant $\epsilon>0$. Is the expected taxation complexity of a truthful deterministic mechanism sampled from the distribution that defines $A$ exponential in $n$ and $m$?
\end{itemize}
\end{open}
 
\subsection{The Taxation Complexity of Truthful in Expectation Mechanisms}\label{subsec-truthful-in-expectation}

We now consider truthful in expectation mechanisms. That is, the mechanism outputs distributions over allocations. The value of player $i$ for an output distribution $\mathcal D$ is $E[v_i(A_i)]$ where $A_i$ is the random variable that denotes the bundle that player $i$ gets in $\mathcal D$. The definition of a truthful mechanism is now the same, with respect to these extensions. Similarly, applying the taxation principle yields that every player is still presented with a menu (the menu now consists of distributions over allocations but still depends only on the valuations of the other players) and is allocated a profit-maximizing distribution.

The next proposition shows that there exists a truthful in expectation mechanism $A$ for general valuations in which $tax(A)>>cc(A)$ (note that by Theorem \ref{thm-indexing-complexity}) if the mechanism is deterministic then $tax(A)\leq cc(A)$). This might hint why in some cases there exist computationally efficient truthful in expectation mechanism that achieve approximation ratios very close to what is possible completely ignoring incentives issues \cite{DD09,DRY11}, whereas we either suspect or know that these approximation ratios are not achievable by computationally efficient deterministic truthful mechanisms.

\begin{proposition}
There is a mechanism $A$ that is \emph{truthful in expectation} for combinatorial auctions with general valuations with $cc(A)=poly(m)$ and $tax(A)=\Omega(exp(cc(A))$.
\end{proposition}
\begin{proof}
Lavi and Swamy \cite{LS05} present a truthful in expectation mechanism that provides an approximation ratio of $O(\sqrt m)$. The communication complexity of the mechanism is $poly(m,n)$. The starting point of their algorithm is the following linear relaxation for maximizing welfare in combinatorial auction:

\vspace{0.1in}\noindent \emph{Maximize:} $\Sigma_{i,S}x_{i,S}v_i(S)$

\vspace{0.1in} \noindent \emph{Subject to:}
\begin{itemize}
\item For each item $j$: $\Sigma_{i,S|j\in S}x_{i,S}\leq 1$
\item for each player $i$: $\Sigma_{S}x_{i,S}\leq 1$
\item for each $i$, $S$: $x_{i,S}\geq 0$
\end{itemize}

The algorithm is maximal in distributional range\footnote{An algorithm is \emph{maximal in distributional range} if there exists a pre-defined range of distributions over allocations and the algorithm always selects the distribution in that range that maximizes the expected welfare. Truthfulness in expectation follows by using VCG payments. See \cite{DD09} for more details.}. In particular, for every set of valuations $v_1,\ldots,v_n$ where the variables of the optimal fractional solution are $\{x_{i,S}\}_{i,S}$ the mechanism outputs the allocation $(S^t_1,\ldots,S^t_n)$ with probability $p^t$ such that $\Sigma_tp^t\cdot \Sigma_iv_i(S^t_i)=\frac { \Sigma_{i,S}x_{i,S}v_i(S)} \alpha$ for $\alpha=\Theta({\sqrt m})$. 

We now give a lower bound on the taxation complexity of that mechanism. We restrict our attention to two players, and show the menu of any two different strictly monotone valuations $v,v'$ of the first player induce a different menu. The proposition will then follow since the number of different strictly monotone valuations is at least doubly exponential\footnote{This can be seen, for example, by considering the set of strictly monotone valuations $v$ such that $v(S)=|S|$ for $|S|\neq \frac m 2$ and $v(S)\in \{\frac m 2,\frac m 2 -1\}$ for $|S|=\frac m 2$. There are $2^{m \choose \frac m 2}$ such valuations.}.

We say that the price of $S$ in the menu presented to player $2$ is the price of the distribution that allocates $S$ to player $2$ with probability $\frac 1 \alpha$ and the empty set otherwise. We will show that for every strictly monotone valuation $v$ of the first player the price of $S$ is $\frac {v(M)-v(M-S)} \alpha$ and the proposition will follow. To see this, consider the following valuation:
$$
u(T)= \begin{cases} 
2v(M) & S\subseteq T ; \\
0 & \text{otherwise.} 
\end{cases} 
$$
The unique optimal fractional solution of the instance $(v,u)$ is $x_{1,M-S}=1$ and $x_{2,S}=1$. Thus the mechanism allocates the bundle $S$ to player $2$ with probability $\frac 1 \alpha$. Since the algorithm is maximal in distributional range the payment (according to the VCG payment scheme) is $\frac {v(M)-v(M-S)} \alpha$.
\end{proof}

\section{Proof of Theorem \ref{thm-indexing-complexity}}\label{app-proof-tax-complexity}

Theorem \ref{thm-indexing-complexity} is proved by a reduction to a new problem that we introduce, the \emph{menu verification problem}.

\paragraph{The Menu Verification Problem.} The menu verification problem is defined with respect to some $n$-player mechanism $A$ which is truthful for some class of valuations $\mathcal V$ and player $i$. Denote by $M^i=\{\mathcal M_k\}_k$ the set of all menus that may be presented to player $i$ in $A$. Let $B= \max_{S:\mathcal M'(S)<\infty,\mathcal M'\in M^i}\mathcal M'(S)$ be the maximum finite price that appear in all menus.

An instance of the menu verification problem is given by $(v_{-i},f)$. Specifically, the input of each player $i'\neq i$ is a valuation $v_{i'}$ and player $i$'s input is a monotone function $f:2^M\rightarrow \mathbb R\cup \{\infty \}$ which is called the \emph{base function}. We assume that $f$ has the following properties: for every $S$ with $f(S)<\infty$ it holds that $f(S)\leq B$ and $f(\emptyset)=0$. Notice that we do not assume that $f\in \mathcal V$.

Let $\mathcal M$ be the menu presented to player $i$ in the mechanism $A$ when the valuations of the other players are $v_{-i}$. The goal in the menu verification problem is to determine whether for some $S$ we have that $f(S) > \mathcal M(S)$. If there is such $S$, the last bit of the protocol should be player $i$ sending the bit $1$, else this last bit has to be $0$.

Recall that the taxation complexity of player $i$ is $\log |M^i|$. The next lemma connects the taxation complexity of $A$ to the communication complexity of the menu verification problem:

\begin{lemma}\label{lemma-verification}
Fix a mechanism $A$ that is truthful for some class of valuations $\mathcal V$. Fix also some player $i$. Suppose that the communication complexity of the menu verification problem (with respect to $\mathcal V$, $A$ and $i$) is $q$. Then, the taxation complexity of player $i$ in $A$ is at most $q-1$.
\end{lemma}
\begin{proof}
We start with a definition:
\begin{definition}[menu generators]
For each menu $\mathcal M\in M^i$, we arbitrarily choose exactly one valuation profile $v^{\mathcal M}_{-i}=(v^{\mathcal M}_1,\ldots,v^{\mathcal M}_{i-1}, \ldots, v^{\mathcal M}_{i+1}, \ldots, v^{\mathcal M}_n)$ of players $(1,\ldots, i-1,i+1,\ldots, n)$ such that the menu $\mathcal M$ is presented to player $i$. The valuation profile $v^{\mathcal M}_{-i}$ is the \emph{menu generator} of $\mathcal M$.
\end{definition}

Consider some menu $\mathcal M\in M^i$. Let $v^{\mathcal M}_{-i}$ be the valuation profile that is a menu generator of $\mathcal M$. The \emph{canonical instance} of a menu $\mathcal M$ in the menu verification problem is the instance $(v^{\mathcal M}_{-i},\mathcal M)$, i.e., the valuations of all players except $i$ are as in $v^{\mathcal M}_{-i}$ and the base function is the menu $\mathcal M$.

Consider a protocol $P$ with communication complexity $q$ for the menu verification problem. Claim \ref{claim-fooling-set} shows that for each $\mathcal M, \mathcal M'\in M^i$, $\mathcal M\neq \mathcal M'$, the transcript of $P$ in the canonical instance of $\mathcal M$ differs from the transcript in the canonical instance of $\mathcal M'$. Recall that the communication complexity of $P$ is $q$, thus the number of transcripts is at most $2^{q}$. This already gives us that the number of canonical instances is at most $2^{q}$. We get a bound of $2^q$ on the number of instances, since the number of canonical instances is exactly $|M^i|$. However, we can do slightly better by observing that although the length of each transcript is $q$, but on each canonical instance the transcript ends with $0$. We get that the number of different transcripts is $2^{q-1}$, which is also a bound on $|M^i|$.

\begin{claim}\label{claim-fooling-set}
For each $\mathcal M, \mathcal M'\in M^i$, $\mathcal M\neq \mathcal M'$, the transcript of $P$ in the instance $(v^{\mathcal M}_{-i}, \mathcal M)$ differs from the transcript in the instance $(v^{\mathcal M'}_{-i}, \mathcal M')$.
\end{claim}
\begin{proof}
Assume towards contradiction that there exist $\mathcal M, \mathcal M'\in M^i$, $\mathcal M\neq \mathcal M'$ such that the transcript of the canonical instance $(v^{\mathcal M}_{-i}, \mathcal M)$ is identical to the transcript of the canonical instance $(v^{\mathcal M'}_{-i}, \mathcal M')$. Using standard fooling-set arguments (e.g., \cite{KN97}), this implies that the transcripts of $(v^{\mathcal M}_{-i}, \mathcal M')$ and $(v^{\mathcal M'}_{-i}, \mathcal M)$ are identical as well. We will show that this is false and reach a contradiction. Towards this end, observe that the last bit that player $i$ sends in both canonical instances is by construction $0$ (since the function $f$ identifies with the menu presented to player $i$). However, we will show that in either $(v^{\mathcal M}_{-i}, \mathcal M')$ or in $(v^{\mathcal M'}_{-i}, \mathcal M)$ the last bit that player $i$ sends is $1$. In particular we get a different transcript, which is a contradiction.

To see that in one of these instances the last bit that is communicated is $1$, notice that since $\mathcal M\neq \mathcal M'$, there must be a bundle $S$ such that $\mathcal M(S)\neq \mathcal M'(S)$. Assume without loss of generality that $\mathcal M(S)> \mathcal M'(S)$. Thus in the instance $(v^{\mathcal M'}_{-i}, \mathcal M)$ the last bit that player $i$ communicates is $1$ (because $f(S)=\mathcal M(S)$).
\end{proof}

This finishes the proof of Lemma \ref{lemma-verification}.
\end{proof}

We are now ready to complete the proof of Theorem \ref{thm-indexing-complexity}. For each class of valuations we will show a protocol $P$ for the menu verification problem with communication complexity equals the specified taxation complexity, and the theorem will follow by applying Lemma \ref{lemma-verification}. 

On a very high level, we would like to run (the most efficient implementation of) $A$ on the instance $(v_{-i},f)$ and decide the menu verification problem by observing the outcome. However, the main obstacle is that $f$ is in general not a proper valuation function (e.g., some entries might be $\infty$) and in particular does not belong to $\mathcal V$. The proof overcomes those obstacles in a different way for each valuation class. 

\paragraph{General valuations.} Given a base function $f$ and player $i$, define this valuation for player $i$:
$$
v_i(S)= \begin{cases} 
 f(S) &  f(S)<\infty; \\
3B & \text{otherwise.} 
\end{cases}
$$
Let $P$ be the following protocol: run the most efficient implementation of $A$ on the instance $(v_{-i},v_i)$. Let $S_i$ be the bundle that player $i$ is allocated. $\mathcal M(S_i)$ is therefore the payment of player $i$. We add an additional step at the end of the protocol: if $v_i(S_i)> \mathcal M(S_i)$ then player $i$ sends $1$ as the last bit, otherwise the last bit player $1$ sends is $0$. 

We claim that the last bit of player $i$ is $1$ if and only if there is some bundle $S$ such that $f(S)> \mathcal M(S)$. We break the proof into two cases, both relying on the simple observation that since $A$ is truthful by the taxation principle $S_i\in \arg\max_S v_i(S)-\mathcal M(S)$. 

\begin{claim}
Suppose that $f(S_i)<\infty$. Then, $f(S_i)> \mathcal M(S_i)$ if and only if for some bundle $S$, $f(S)> \mathcal M(S)$.
\end{claim}
\begin{proof}
As noted above, $S_i\in \arg\max_S v_i(S)-\mathcal M(S)$. When $f(S_i)<\infty$ it holds that $f(S_i)=v_i(S_i)$ by construction, and thus if $v_i(S_i)>\mathcal M(S_i)$ then in particular $f(S_i)>\mathcal M(S_i)$. If $v_i(S_i)=f(S_i)\leq \mathcal M(S_i)$ then, since $S_i$ maximizes the profit, for all $S$ such that $f(S)<\infty$ we have that $f(S)\leq \mathcal M(S)$. As for $S$ with $f(S)=\infty$, if $\mathcal M(S)=\infty$ then $f(S)\leq \mathcal M(S)$. To finish the proof we claim that for all $S$ with $f(S)=\infty$, it holds that $\mathcal M(S)=\infty$. Otherwise, the profit of $S$ is $v_i(S)-\mathcal M(S)\geq 3B-B=2B$. The profit of $S_i$ is $v_i(S_i)-\mathcal M(S_i)\geq B-0$. Thus, the profit from the bundle $S$ is strictly larger than the profit from bundle $S_i$, a contradiction.
\end{proof}

\begin{claim}
Suppose that $f(S_i)=\infty$. There is some bundle $S$ for which $f(S)> \mathcal M(S)$.
\end{claim}
\begin{proof}
Since $S_i$ is the bundle that player $i$ is assigned in $A$, it must be that $\mathcal M(S_i)<\infty$. Hence we have that $f(S_i)>\mathcal M(S_i)$.
\end{proof}

\paragraph{Subadditive valuations.} The proof is almost identical to the proof for general valuations. The only change is that instead of considering the valuation $v_i$ we consider the subadditive valuation $v'_i(S)=v_i(S)+\max_Tv_i(T)$ for all $S\neq \emptyset$ (for $S=\emptyset$ we set $v'_i(S)=0$ as usual). Player $i$ sends $1$ if and only if $v'_i(S_i)-\max_Tv_i(T)>\mathcal M(S)$. The rest of the proof is identical since the proof relies only on the difference between pairs of bundles and for every $S$ and $S'$ ($S,S'\neq \emptyset$): $v_i(S)-v_i(S')=v'_i(S)-v'_i(S')$.

\paragraph{XOS valuations.} Fix some bundle size $r$ ($1\leq r\leq m$). We describe a protocol with communication complexity $c+1$ that will determine whether there exists some bundle $S$, $|S|=r$, such that $f(S)>\mathcal M(S)$. The claim will then follow by observing that there are $m$ possible bundle sizes. Define an XOS valuation $v_i$ that consists of the following additive valuations: for every bundle $T$ with $f(T)<\infty$, $|T|=r$, let $a_T$ be the valuation where $a_T(\{j\})= \frac {f(T)} {r} +3B$. If $f(T)=\infty$ and $|T|=r$, let $a_T$ be the valuation where $a_T(\{j\})= \frac {2B} {r} +3B$. 

Now, for every one of the possible bundle sizes $r$, run the most efficient implementation of $A$ on the instance $(v_{-i},v_i)$. Let $S_i$ be the bundle that player $i$ is allocated and $\mathcal M(S_i)$ is the payment of player $i$. We add an additional step at the end of $A$: if $|S_i|<r$ then the last bit that player $i$ sends is $0$. Else, if $v_i(S_i)-3B\cdot r> \mathcal M(S_i)$ then player $i$ sends $1$, otherwise he sends $0$. The protocol ends when player $i$ sends $1$ if and only if at the end of at least one step he sent $1$. Otherwise, the final bit is $0$.

The bound on the communication complexity for XOS valuations  follows from the next lemmas:

\begin{claim}\label{claim-xos-size}
Let $S_i$ be the bundle that player $i$ is assigned in the round where bundles of size $r$ are considered. If $|S_i|<r$ then there is no bundle $S$, $|S|=r$ with $f(S) > \mathcal M(S)$. Otherwise, there exists some $S'_i$ (possibly $S'_i=S_i$) such that $|S'_i|=r$, $f(S_i)=f(S'_i)$ and $v_i(S_i)=v_i(S'_i)$.
\end{claim}
\begin{proof}
Suppose first that $|S_i|<r$. The profit is $v_i(S_i)-\mathcal M(S_i)= (\frac {f(S_i)} {r} +3B)\cdot |S_i|-\mathcal M(S_i)\leq B+ (r-1)\cdot 3B$. Suppose that there exists some bundle $T$ such that $|T|=r$ and $\mathcal M(T)<\infty$. The profit from $T$ is at least $|T|\cdot 3B-\mathcal M(T)\geq r\cdot 3B-B$, where we use the fact that the maximum price in $\mathcal M$ is $B$. Thus the profit from $T$ is strictly larger than the profit $S_i$, which is a contradiction to the taxation principle, since $S_i$ is not a profit-maximizing bundle. Therefore, if $|S_i|<r$ then there is no bundle $S$, $|S|=r$ with $f(S) > \mathcal M(S)$, since $\mathcal M(S)=\infty$ for all $S$ with $|S|=r$.

Now suppose that $|S_i|>r$. Let $T=\arg\max_{T\subseteq S_i, |T|=r}v_i(T)$. Notice that since $f$ is monotone and since by the taxation principle the profit of $S_i$ is at least that of $T$, we have that $v_i(T)\geq v_i(S_i)$. However, notice that by construction in the maximizing clause of $T$ there are only $r$ non zero elements, hence it holds that $v_i(T)= v_i(S_i)$, as needed.
\end{proof}

\begin{claim}
Let $S'_i$ be the bundle guaranteed by Claim \ref{claim-xos-size}. Suppose that $f(S'_i)<\infty$. Then, $f(S'_i)-3B\cdot r> \mathcal M(S'_i)$ if and only if for some bundle $S$, $f(S)> \mathcal M(S)$.
\end{claim}
\begin{proof}
Since $S_i\in \arg\max_S v_i(S)-\mathcal M(S)$, by Claim \ref{claim-xos-size} $S'_i\in \arg\max_S v_i(S)-\mathcal M(S)$. When $f(S'_i)<\infty$ it holds that $f(S'_i)+3B=v_i(S'_i)$ by construction, and thus if $v_i(S'_i)-3B>\mathcal M(S'_i)$ then in particular $f(S_i)>\mathcal M(S_i)$. If $v_i(S'_i)=f(S'_i)\leq \mathcal M(S'_i)$ then, since $S'_i$ maximizes the profit, for all $S$ such that $f(S)<\infty$ we have that $f(S)\leq \mathcal M(S)$. As for $S$ with $f(S)=\infty$, if $\mathcal M(S)=\infty$ then $f(S)\leq \mathcal M(S)$. To finish the proof of the claim we show that for all $S$ with $f(S)=\infty$, we have that $\mathcal M(S)=\infty$. Assume otherwise. The profit of $S$ is $v_i(S)-\mathcal M(S)\geq 5B-B=4B$. The profit of $S_i$ is $v_i(S_i)-\mathcal M(S_i)\geq 3B-0$. Thus, the profit of $S$ mis strictly bigger than that of $S_i$, a contradiction to the taxation principle.
\end{proof}

\begin{claim}
Let $S'_i$ be the bundle guaranteed by Claim \ref{claim-xos-size}. Suppose that $f(S'_i)=\infty$. There is some bundle $S$ for which $f(S)> \mathcal M(S)$.
\end{claim}
\begin{proof}
Since $S_i$ is the bundle that player $i$ is assigned in $A$, it must be that $\mathcal M(S_i)=\mathcal M(S'_i)<\infty$. Hence we have that $f(S'_i)>\mathcal M(S'_i)$.
\end{proof}

\paragraph{Submodular valuations.} We will develop a protocol with communication complexity $c+1$ that determines whether there exists some bundle $S$, $|S|=k$ and $f(s)=w$ such that $f(S)>\mathcal M(S)$. The claim will then follow by observing that there are $m$ possible bundle sizes and $d$ possible values. Let $\mathcal S=\{|S|=k \hbox{ and } f(s)=w\}$. Let $t=2^{m+1}\cdot B$. Define the following valuation:
\[
v_i(S)=\left\{
  \begin{array}{ll}
    |S|\cdot t, & |S|<k; \\
    k\cdot t, & \hbox{$\exists T\in\mathcal S$} \hbox{ s.t. $T\subsetneq S$} ; \\
    ( k-\frac 1 {2^{|S|}}) \cdot t, & \hbox{otherwise.}
  \end{array}
\right.
\]

\begin{claim}
\label{cl:v-submodular}
$v_i$ is non-decreasing and submodular.
\end{claim}
\begin{proof}
It is easy to verify that $v_i$ is non-decreasing. We now show that $v_i$ is submodular, i.e., marginal values are non-increasing: $v_i(S\cup \{j\})-v_i(S)\leq v_i(T\cup \{j\})-v_i(T)$, for every $T\subsetneq S, j\notin S$. We divide the analysis into several cases:
\begin{itemize}
\item $|T| \leq k-2$: then we have $v_i(T\cup \{j\})-v_i(T)=t$. We always have $v_i(S\cup \{j\})-v_i(S) \leq t$, because all marginal values are at most $t$.

\item $|T| = k-1$: then $v_i(T\cup \{j\}) - v_i(T) \geq \frac{t}{2}$, since all bundles of size $k$ have value at least $(k-\frac12) t$. On the other hand, $v_i(S\cup \{j\}) - v_i(S) \leq \frac{t}{2}$, because by the same token $v_i(S) \geq (k-\frac12) t$ and $v_i(S \cup \{j\}) \leq kt$.

\item $|T| \geq k$ and some subset of $T$ (possibly itself) is in $\mathcal S$: in this case, it is enough to note that $v_i(S) = k\cdot t$ by the second case of the definition, and so $v_i(S\cup \{j\}) - v_i(S) = 0$, while $v_i(T\cup \{j\}) - v_i(T) \geq 0$.

\item $|T| \geq k$ and no subset of $T$ (including itself) is in $\mathcal S$:
Then $v_i(T) = (k - \frac{1}{2^{|T|}}) t$ by the last case of the definition,
and $v_i(T\cup \{j\}) - v_i(T) \geq \frac{t}{2^{|T|+1}}$. On the other hand, $v_i(S) \geq (k - \frac{1}{2^{|S|}}) t$, which implies that $v_i(S\cup \{j\}) - v_i(S) \leq \frac{t}{2^{|S|}} \leq \frac{t}{2^{|T|+1}}$.
\end{itemize}
\end{proof}

Now, for every possible bundle size $k$ and value $w$, run the most efficient implementation of $A$ on the instance $(v_{-i},v_i)$. Let $S_i$ be the bundle that player $i$ is allocated and $\mathcal M(S_i)$ be the payment of player $i$. We add an additional step at the end of $A$: if $v_i(S_i)=t\cdot k$ and $\mathcal M(S_i)< w$ then player $i$ sends $1$, otherwise he sends $0$. The protocol ends when player $i$ sends $1$ if and only if at the end of at least one step he sent $1$. Otherwise, the final bit is $0$.

\begin{claim}
Let $S_i$ be the bundle that player $i$ is assigned in the round where we consider bundle size $k$ and value $w$.
If there exists a bundle $S$ such that $f(S)=w$ and $\mathcal M(S)<\infty$, then $v_i(S_i)=t\cdot k$.
\end{claim}
\begin{proof}
Suppose that there exists some bundle $S$ such that $f(S)=w$ and $\mathcal M(S)<\infty$. The profit from $S$ is at least $v_i(S)-\mathcal M(S)=t\cdot k-B$. We now compute the profit of any other bundle $T$ with $v_i(T)<t\cdot k$:  $v_i(S_i)-\mathcal M(S_i)\leq v_i(S_i)-0  \leq t\cdot k - \frac t {2^m}=t\cdot k-\frac {2^{m+1}\cdot B} {2^m}=t\cdot k-2B$. Thus, $S$ is more profitable than any such $T$ and by the taxation principle $T$ cannot be selected.
%
\end{proof}

Let $\mathcal S=\{S| f(S)=w \hbox { and }\mathcal M(S)<\infty\}$. If $\mathcal S\neq \emptyset$ then by the claim $S_i\in \mathcal S$ and furthermore by the taxation principle $\mathcal M(S_i)\in \arg\min_{S\in \mathcal S}\mathcal M(S)$. In particular, if there exists a bundle $S\in \mathcal S$ with $\mathcal M(S)<w$, then $\mathcal M(S_i)<w$.

If $\mathcal S=\emptyset$, then it immediately holds that for all bundles $S$ with $f(S)=w$ it holds that $\infty =\mathcal M(S)>f(S)=w$. This finishes the proof of theorem \ref{thm-indexing-complexity}.

\section{Proof of Theorem \ref{thm-construct-menu} -- The Menu Reconstruction Theorem}\label{sec-thm-construct-menu}

Assume that the player $i$ that we want to construct the menu for is $n$ (the proof is identical otherwise). Fix the valuations of the other players to be $v_{-n}=(v_1,\ldots, v_{n-1})$ and let $\mathcal R$ be the menu they present to player $n$. The idea is to find $\mathcal R$ by obtaining for every menu $\mathcal M\in M^i$, $\mathcal M\neq \mathcal R$, a ``witness'' that proves that $\mathcal M\neq \mathcal R$. Specifically, a menu $\mathcal M$ is \emph{not alive} if we have found some bundle $S$ such that $\mathcal M(S)\neq \mathcal R(S)$. 

The basic idea is to have several \emph{shrinkage} steps, where in each step $j$ we construct a set $U_j\subseteq U_{j-1}$ of menus that are still alive. Let $U_0=M^n$ (where $M^n$ is the set of all menus that may be presented to player $n$) and recall that $tax(A)\geq \log |M^n|$. The goal of each shrinkage step $j$ is to find a set of live menus $U_{j}$ such that $|U_{j}|\leq \frac {|U_{j-1}|} 2$ using only $poly(tax(A),price(A))$ bits of communication, obviously making sure that $\mathcal R\in U_j$. If we are able to accomplish that, we are guaranteed that after at most $tax(A)$ steps only one menu is alive. This menu must be $\mathcal R$.

\subsection{Part 1: Shrinkage Steps}\label{subsec-shrinkage}

We now describe shrinkage step $j$. For every bundle $S$, let $p_S\in\arg\max_{p}|\{\mathcal M:\mathcal M(S)=p,\mathcal M\in U_{j-1}\}|$ be a most frequent price of $S$ among all menus that are still live. If $p_S$ repeats in less than half of the menus that are in $U_{j-1}$, we check the price $\mathcal R(S)$ (using $price(A)$ bits of communication) and shrink the set of live menus: keep in $U_j$ only menus $\mathcal M\in U_{j-1}$ with $\mathcal M(S)= \mathcal R(S)$. 

The more difficult case is when for every bundle $S$, $p_S$ repeats in at least half of the menus of $U_{j-1}$. The solution to this case relies on a reduction to a problem that we call \emph{$z$-disjointness}.

\paragraph{$z$-Disjointness.} 

The basic setup of the $z$-disjointness problem is very similar to that of (the multi-player version of) set disjointness. We have $n$ players, where the input of each player $i$ is $A^i\in \{0,1\}^l$. A bit $k$ is \emph{disjoint} if there exists some player $i$ with $A^i=0$, bit $k$ \emph{intersects} otherwise.  The goal, as in the set-disjointness problem, is to determine whether there is a bit that intersects. Additionally, we have the following restrictions:
\begin{itemize}
\item For each player $i$ the input $A^i$ of player $i$ belongs to some (known) set $S^i\subseteq \{0,1\}^l$.

\item For every possible input in $(A_1,\ldots, A_n)\in S^1\times \ldots \times S^n $ there are at most $z$ bits that intersect.
\end{itemize}

We will show that:
\begin{theorem}\label{thm-block-disjointness}
There is protocol with communication complexity $O(z^2\cdot n^2\cdot \log l)$ that decides the $z$-disjointness problem.
\end{theorem}
We postpone the proof of Theorem \ref{thm-block-disjointness} to Section \ref{sec-solve-block-disjoitness}.

\subsection{Part 2: ``Decreasing'' the Number of Witness Bundles}

Our goal is to reduce the problem of determining whether there is a bundle $S$ with $\mathcal R(S)\neq p_S$ (a witness bundle) to the $z$-disjointness problem. It will be helpful to ``reduce'' the number of witness bundles to $z$ (for a value of $z$ that will be determined later). For a menu $\mathcal M$, denote by $w(\mathcal M)$ the number of its witness bundles.

\begin{definition}
Let $Z'$ be a set of menus such that for each $\mathcal M\in Z$, $w(\mathcal M)\leq z$. Let $Z\subseteq Z'$ be a set of menus such with $\frac z 2 \leq w(\mathcal M)\leq z$. We say that a set of bundles $P$ \emph{represents} $Z$ if the following conditions simultaneously hold:
\begin{enumerate}
\item For every $\mathcal M\in Z$, there is a bundle $S\in P$ such that $\mathcal M(S)\neq p_S$.
\item For every $\mathcal M\in Z'$, $w(\mathcal M)\leq 8\log |Z'|$.
\end{enumerate}
\end{definition}

\begin{lemma}\label{lemma-represents}
There is a set of bundles $P$ that represents $Z$.
\end{lemma}
\begin{proof}
We show that if $P$ is constructed by selecting each bundle uniformly at random with probability $\frac {4\log |Z'|} {z} $ then with high probability $P$ represents $Z$. Fix some menu $\mathcal M\in Z$. Since $\mathcal M$ has at least $\frac z 2$ witnesses, the probability that there is no bundle $S\in P$ that is a witness for $\mathcal M$ is at most $(1-\frac 1 {\frac {4\cdot \log |Z'|} {z}})^{\frac {z} 2}\leq \frac 1 {e^{2\log |Z'|}}$. Using the union bound, the probability that $P$ does not contain a witness bundle for some $\mathcal M\in Z$ is at most $|Z| \cdot \frac {1} {e^{2\log |Z'|}}\leq \frac 1 {|Z'|}$.

We now show that the second property is not violated with high probability. Fix some menu $\mathcal M\in Z'$. The number of witness bundles for $\mathcal M$ is at most $z$, and thus the expected number of witnesses for $\mathcal M$ in $P$ is at most $\frac {4\log |Z'|} {z}\cdot z=4\log |Z'|$. By the chernoff bounds, the probability that the number of witnesses for $\mathcal M$ that are in $P$ is greater than $8\cdot \log |Z'|$ is at most $e^{-\frac {4 \log |Z'|} 3}$. There are at most $|Z'|$ menus. Therefore, by the union bound, the probability that the second property is violated is at most $|Z'|\cdot e^{-\frac {4 |\log Z'|} 3}\leq e^{-\frac {|\log Z'|} 3}$.

To conclude, the probability that either of the properties that are necessary for $P$ to represent $Z$ is violated is at most $\frac 1 {e^{\log |Z'|}}+ \frac 1 {e^{\frac {\log |Z'|} 3}}<1$. Thus, there is a set $P$ that represents $Z$.
\end{proof}

We can now describe our high-level approach:
\begin{enumerate}
\item Let $Z'=U_{j-1}$.
\item For each $t=2^m,2^{m-1}, \ldots, 4,2,1$, in decreasing order:
\begin{enumerate}
\item Let $Z=\{\mathcal M|\frac t 2 \leq w(\mathcal M)\leq t\}$.
\item Obtain a set of bundles $P$ that represents $Z$.
\item Using $P$, determine whether $\mathcal R\in Z$ or not by finding a witness. This is done by a reduction to $z$-disjointness that will be shortly described.
\item If a witness for $\mathcal R$ was found, perform a shrinkage step. Else, let $Z'=Z'-Z$ and continue to the next value of $t$.
\end{enumerate}
\end{enumerate}

The logic is as follows: for each value of $t$, in descending order, we ``guess'' that $\frac t 2 \leq w(\mathcal R)\leq t$. We then construct an instance of $z$-disjointness that contains an intersecting bit if and only if our guess was correct. We then solve the $z$-disjointness instance. If we have found an intersecting bit then we have found a witness bundle and we can perform a shrinkage step. Otherwise, we remove from $U_{j-1}$ all menus $\mathcal M$ with $\frac t 2 \leq w(\mathcal M)\leq t$ and proceed similarly. If the process ends without finding a witness bundle, we can conclude that $\mathcal R(S)=p_S$ for every bundle $S$.

Our reduction to $z$-disjointness will use $z=O(\log|Z'|)=O(tax(A))$ and $l=2^{m+price(A)}$. Thus, the number of bits needed to solve each instance (by Theorem \ref{thm-block-disjointness}) is $poly(tax(A), price(A), m, n)$. We have to solve at most $m$ $z$-disjointness problems (one for each possible value of $t$), so the communication complexity of a shrinkage step is $poly(tax(A), price(A), m, n)$. We now describe the details of the reduction.

\subsection{Part 3: A Reduction to $z$-Disjointness}

We now fill in the details of our high-level approach by describing the reduction to $z$-disjointness. The reduction is inspired by the reduction of \cite{BFS86} which shows that disjointness is co-NP complete (in the communication sense). 

\begin{definition}
Fix a truthful mechanism $A$ and a protocol $Q$ for computing the price of bundle $S$ in the menu that is presented to player $n$. Let $price(A)$ denote the communication complexity of $Q$. A communication transcript $T$ is a \emph{proof} for bundle $S$ and player $i<n$ with valuation $v_i$, if there exist valuations $v'_{-i}$ of all players $1,\ldots, i-1,i+1,\ldots, n-1$, such that when running $Q$ on the instance $(v_i,v'_{-i})$ the communication transcript is $T$ and the price of $S$ in the menu presented to player $n$ is not $p_S$.
\end{definition}
The reduction takes a set of bundles $P$ and defines a unique ``block'' that corresponds to each bundle $S\in P$. The number of bits in a block is $2^{price(A)}$, where each bit in the block corresponds to a different possible transcript of $A$. Therefore, the total length of the input string $A^i$ in the $z$-disjointness problem is $|P|\cdot 2^{price(A)}$. We set bit $k$ in player $i$'s input string that is in the block that corresponds to a bundle $S$ to $1$ if and only if the transcript that corresponds to $A^i_k$ is a proof for bundle $S$ and player $i$. The next claim is the crux of the reduction:

\begin{claim}
Fix some bit $k$ that is in the block that corresponds to bundle $S$. Then, we have that $A^1_k=\ldots=A^{n-1}_k=1$ if and only if $\mathcal R(S)\neq p_S$. Furthermore, in every block there is at most one bit that intersects.
\end{claim}
\begin{proof}
Suppose first that $A^1_k=\ldots=A^{n-1}_k=1$. We will show that for every player $i$, $v_{-i}$ is a proof for player $i$ and the bundle $S$, and thus $\mathcal R(S)\neq p_S$. Towards this end, let $v^i_{-i}$ be valuations that induce a proof for player $i$ and $S$ (notice that by assumption such a proof exists). Notice that the transcripts of the instances $(v_1,v^1_{-1})$ and $(v_2,v^2_{-2})$ are identical by definition. By standard fooling set arguments (see, e.g., \cite{KN97}), the transcript of the instance $(v_1,v_2,v^1_3,\ldots,v^1_{n-1})$ is identical as well. Repeating the same argument by replacing the valuations of players $v_3,\ldots, v_{n-1}$ one by one we get that the transcript of the instance ($v_1,\ldots, v_{n-1})$ is the same. We therefore have that for every player $i$, $v_{-i}$ is a proof for player $i$ and $S$, and thus $\mathcal R(S)\neq p_S$.

In the other direction, if for some player $i$ we have that $A^i_k=0$ then by the reduction there are no valuations of the other players such that the corresponding transcript is a proof for player $i$ and $S$. Therefore, if for every bit $k$ in the block we have some player $i$ with $A^i_k=0$ then $\mathcal R(S)=p_S$. 

To see that every block contains at most one intersecting bit, observe that if bit $k$ in the block $S$ intersects then the transcript $T$ it corresponds to is the transcript of running $Q$ to compute $S$ when the valuations of the players are $v_1,\ldots, v_{n-1}$. Thus there can be at most one bit that intersects.
\end{proof}

Notice that since we use $P$ that represents $Z$, there are at most $8\log |Z'|$ bundles $\mathcal R(S)\neq p_s$. By the last claim, this is also a bound on intersecting bits in every possible instance. Thus, we can apply Theorem \ref{thm-block-disjointness} and decide whether there is a bundle $S$ with $\mathcal R(s)\neq p_S$ whenever $\mathcal R\in Z$.

\subsection{Proof of Theorem \ref{thm-block-disjointness}: Solving $z$-Disjointness}\label{sec-solve-block-disjoitness}

An efficient protocol for solving $z$-block disjointness can be obtained as a corollary from the literature on the communication complexity of problems with a small number of witnesses \cite{KNSW94}. For completeness and since \cite{KNSW94} only handles the two-player case we bring the explicit proof here. We note that this is the only part of the proof of Theorem \ref{thm-construct-menu} that uses non-standard queries, even if the original mechanism uses only standard queries.

We start with solving $z$-disjointness for the special case when $z=1$. We will then use this solution and solve for the case of general $z$. 

\begin{lemma}\label{lemma-one-restricted}
The communication complexity of solving $1$-disjointness is $O(n^2\cdot \log l)$.
\end{lemma}
\begin{proof}
We provide a multi-round protocol. In each round $s$ we maintain a set $K_s\subseteq K_{s+1}$ of bits such that for every $k\in K_s$ there exists at least one player $i$ with $A^i_k=0$. More specifically, initially we let $K_0=\emptyset$. Then, in each round we either decide the problem by finding a bit $k$ such that $A^1_k=\ldots =A^{n}_k=1$, or obtain a set $K_{s}$ by adding to the set $K_{s-1}$ at least $\frac {l-|K_{s-1}|} {4n}$ additional bits. We let $r_s=l-|K_{s-1}|$.

Crucially, if $k\in K_s$ then it cannot be the case that $A^1_k=\ldots =A^{n}_k=1$, thus we can ``ignore'' this bit. In other words, we can simply ``discard'' for each player $i$ and string $A\in S^i$ every bit $k$ such that $k\in K_s$ (by setting $A_k=0$) and effectively obtain a smaller problem with significantly less bits to consider. We can now proceed to the next round. After $O(n\cdot \log l)$ rounds we will either find a common $1$ bit or will be left with only a constant number of bits which may contain a common $1$ bit. In the latter case, all players report the value of all bits that are not in $K_s$ and determine whether there is a bit that intersects.

Consider some round $s$. Recall that $r_s$ is essentially the number of bits in the instance of $1$-disjointness that we aim to solve in this round. We need some definitions. The \emph{neighborhood} of the $k$'th bit of player $i$ is $N_i(k)=\{j|A_j=1,A_k=1, A\in S^i\}$. If $|N_i(k)|>(1-\frac 1 {2n})r_s$ the neighborhood is \emph{large}, otherwise it is \emph{small}.


Next, each player $i$ reports the index of some bit $k_i$ with a small neighborhood such that $A^i_{k_i}=1$, if such exists. This takes $n\cdot \log l$ bits of communication. Observe that if there is an intersecting bit, then it must be in the neighborhood of $k_i$. Suppose that there exists a player $i$ who reported $k_i$ with a small neighborhood (we will shortly handle the case where no such player exists). Observe that if there is an intersecting bit it must reside in $N_i(k_i)$. If $N_i(k_i)$ is small then, since $N_i(k_i)$ is known to all players, we can progress to the next step as follows: we obtain $K_s$ by adding to $K_{s-1}$ all bits that are not in $N_i(k_i)$. Observe that $|K_s|\geq |K_{s-1}|+ \frac 1 {2n}\cdot r_{s-1}$.

If none of the players has a bit $k_i$ with $A^i_{k_i}=1$ and a small neighborhood, then there is no intersecting bit:
\begin{claim}
Suppose that $A^1_k=\ldots=A^n_k=1$ for some bit $k$. Then there is at least one player $i$ where $N_i(k)$ is small.
\end{claim}
\begin{proof}
Let $T=\cap_{i=1}^{n}N_{i}(k)$. For each $i$, $|N_i(k)|\geq (1-\frac 1 {2n})r_s$ since it is large. There are $n$ players and thus $|\cap_{i=1}^{n}N_{i}(k)|> \frac {r_{s-1}} 2$. Let $k'\neq k$ be some bit in $T$. We construct an instance with at least two intersecting bits: by definition of neighborhood for every $i$ there is some $C^{i}\in S^{i}$ where $C^{i}_{k}=C^{i}_{k'}=1$. Thus, in the instance $(C^1,C^2,\ldots, C^{n})$ we have that $C^1_k=C^2_k=\ldots=C^n_k=1$ and $C^1_{k'}=C^2_{k'}=\ldots=C^{n}_{k'}=1$. We got an instance with at least two intersecting bits, contradicting our promise.
\end{proof}

This concludes the proof of Lemma \ref{lemma-one-restricted}.
\end{proof}

\noindent The next lemma relates the communication complexity of $1$-disjointness to that of $z$-disjointness.

\begin{lemma}\label{lemma-restricted-reduction}
Let $f_z(l)$ denote the communication complexity of solving any $z$-disjointness problem with $n$ players on $l$-bit strings. Then, for any $z\geq 2$, $f_z(l)\leq f_{1}(l^z)+f_{z-1}(l)$. In particular, we get that $f_z(l)\leq z\cdot f_1(l^z)$.
\end{lemma}
\begin{proof}
The proof is by induction on $z$, starting with $z=2$. Consider a string $A^i$ of length $l$. For an integer $z>1$, define the \emph{$z$-product} of $A^i$ to be a string $s$ that is constructed as follows: $s$ consists of ${l\choose z}$ bits, one for each tuple of $z$ different bits. We set a bit to $1$ if and only if all the $z$ bits that it corresponds to are $1$ in $A^i$.

For each player $i$, apply a $z$-product to every $A^i\in S^i$. We obtain a new problem. Let $(A^1,\ldots, A^n)$ be some instance of the original problem and $(A'^1,\ldots, A'^n)$ be some instance of the new problem where for each $i$, $A'^i$ is the $z$-product of $A^i$. Observe that $(A^1,\ldots, A^k)$ contains exactly $z$ intersecting bits if and only if $(A'^1,\ldots, A'^k)$ contains exactly one intersecting bit. In particular, if in the original problem the number of intersecting bits is at most $z$, then after applying the $z$-product there is at most one intersecting bit. Since by applying a $z$-product we get strings of length at most $l^z$, the communication complexity of the new problem is $f_1(l^z)$. We have therefore established that for every instance $(A^1,\ldots, A^n)$ deciding whether the number of intersecting is exactly $z$ takes at most $f_1(l^z)$ bits of communication.

Given an instance $(A^1,\ldots, A^{n})$ we can run the protocol above. If the protocol finds an intersecting bit, we are of course done. Else, the communication protocol has reached some leaf which is labeled ``no''. Recall that any leaf in the communication protocol corresponds to some sets $S'^1,\ldots , S'^k$ where each instance $(A^1,\ldots, A^n)$ with $A^i\in S'^i$ generates the same transcript. In particular, any such instance contains at most $z-1$ intersecting bits. Using the induction hypothesis we can solve this instance using $f_{z-1}(l)$ bits of communication.
\end{proof}

Thus, combining Lemmas \ref{lemma-one-restricted} and \ref{lemma-restricted-reduction} gives us Theorem \ref{thm-block-disjointness}.

\section{Some Aspects of Theorem \ref{thm-characterization}} 

Here we discuss several aspects of the characterization of the communication complexity of mechanisms that was provided in Theorem \ref{thm-characterization}. First, in Subsection \ref{subsec-price-computation} we prove that $price(A)\leq cc(A)$ for domains that contain additive valuations. Then, we show the tightness of our characterization by showing that all terms $tax(A), price(A), tie(A)$ are needed (Subsection \ref{sec-tight-char}).

\subsection{Computing the Price of a Bundle}\label{subsec-price-computation}

The next proposition shows that if the domain contains additive valuations then $price(A)\leq cc(A)$.

\begin{proposition}\label{proposition-price-of-bundle}
Consider a mechanism $A$ that is truthful for a domain of valuations that contains additive valuations. Then, $price(A)\leq cc(A)$.
\end{proposition}
\begin{proof}
Fix some player $i$ and let $v_{i'}$ be the valuation of every player $i'\neq i$. Denote by $M^i=\{\mathcal M_k\}_k$ the set of all menus that may be presented to player $i$ in the mechanism $A$. Let $B= \max_{S:\mathcal M'(S)<\infty,\mathcal M'\in M^i}\mathcal M'(S)$ be the maximum finite price that appear in all menus.

Let $\mathcal M$ be the menu that $v_{-i}$ present to player $i$. To compute $\mathcal M(S)$, consider the following additive valuation $v_i$ of player $i$: $v_i(\{j\})=3B$ for $j\in S$ and $v_i(\{j\})=0$ otherwise. Run the most efficient implementation of $A$ on the instance $(v_{i},v_{-i})$ and let $S_i$ be the bundle allocated to player $i$ and $p$ be its payment. We divide the proof into several cases.

If $S_i=S$ then obviously the price of $S$ in the menu is $p$ and we are done. If $S\subseteq S_i$ then we observe that since by construction $v_i(S_i)=v_i(S)$ and the menu price of $S_i$ is at most that of $S$ (by the monotonicity of the menu ), it must be that the menu price of $S_i$ equals that of $S$, otherwise $S_i$ does not maximize the profit. Thus the price of $S$ in the menu is $p$ as well.

We will show that in all other cases the price of $S$ in the menu is $\infty$. If $|S_i|<|S|$ then the profit of $S_i$ is at most $|S_i|\cdot 3B\leq (|S|-1)\cdot 3B$, whereas had the price of $S$ been finite the profit were $|S|\cdot 3B-B$, in contradiction to the taxation principle, since $S_i$ is a profit-maximizing bundle. Thus the price of $S$ in the menu must be infinite in this case. 

We have already handled the case where $S\subseteq S_i$, so the only remaining case to consider is when $|S_i|\geq |S|$ but $S$ is not contained in $S_i$. Then, similarly to before, the profit from $S_i$ is at most $(|S|-1)\cdot 3B$ whereas the profit from $S$ is at least $|S| \cdot 3B-B$, which is again a contradiction to the taxation principle.
\end{proof}

If the domain does not contain additive valuations, we do not know the communication complexity of computing the price of player $i$ for a given alternative:
\begin{open}
Let $A$ be a truthful mechanism for an arbitrary domain. Is $price(A)\leq poly(cc(A))$?
\end{open}

\subsection{Tightness of the Characterization}\label{sec-tight-char}

Fix a mechanism $A$ that is truthful for general valuations. Informally, Theorem \ref{thm-characterization} gives us that $cc(A)=poly(tax(A),price(A), tie(A),n,m)$. In this section we show that all terms are necessary in the sense that omitting just one of $tax(A)$, $price(A)$, or $tie(A)$ might imply a huge gap between the LHS and the RHS.

All of our examples are based on (different) reductions from set disjointness. In this problem, Alice receives a string $A$ and Bob receives a string $B$, $A,B\in \{0,1\}^t$. The goal is to determine whether there is some index $k$ such that $A_k=B_k=1$. A simple fooling set argument shows that the deterministic communication complexity of $f$ is exactly $t$ (a similar bound holds for randomized mechanisms, see, e.g., \cite{R92}). We note that it is straightforward to extend these examples and obtain analogous results for mechanisms that use only value queries.

\subsubsection{Dropping $tie(A)$}\label{subsec-low-tax-high-comm}

We show that $tie(A)$ is necessary to determine the communication complexity of a truthful mechanism $A$.

\begin{proposition}
There is a truthful mechanism $A$ with $tax(A)=1$ (and thus $price(A)=0$) and $cc(A)=exp(m,n)$.
\end{proposition}
\begin{proof}
Let $f$ be the following two-player social choice function: player $1$ is never allocated any items. Player $2$ with valuation $v_2$ gets item $a$ if $v_2(\{a\})>v_2(\{b\})$ and item $b$ if $v_2(\{b\})>v_2(\{a\})$. If $v_2(\{b\})=v_2(\{a\})$ then player $2$ gets either $a$ or $b$ according to some tie-breaking rule that will be specified shortly. Player $1$ is never allocated any items. Notice that $f$ can obviously be implemented truthfully and that its taxation complexity is $1$ (the price for player $2$ for items $a$ and $b$ is $0$ and the price of the rest of the bundles is $\infty$).

We now describe the tie-breaking rule. If there exists some player $i$ and bundle $S$ with $v_i(S)\notin \{0,1\}$ then player $2$ is allocated item $a$. Else, player $2$ is allocated item $a$ if there exists some bundle $S$, $|S|=m/2$, with $v_1(S)=v_2(S)$. Player $2$ is allocated item $b$ otherwise.

We claim that the communication complexity of any mechanism that implements $F$ is at least $exp(m)$. To prove this, we reduce from the set-disjointness problem. Let $t={m \choose {\frac m 2}}$. Let the valuation of Alice (player 1) be identically $0$ for bundles $S$ with $|S|<\frac m 2$ and identically $1$ for bundles $S$ with $|S|>\frac m 2$. For bundles $S$ with $|S|=\frac m 2$ the value of $S$ is determined by the value of $A_S$, where we assume some one-to-one and onto correspondence between the indices of $A$ and bundles of size $\frac m 2$. Bob's valuation (player 2) is defined similarly.

Notice that if one can implement $f$ once can also solve set-disjointness. The communication complexity of $f$ is therefore ${m \choose {\frac m 2}}=exp(m)$.
\end{proof}

\subsubsection{Dropping $tax(A)$}

We give an example for a truthful mechanism $A$ for two players and $m$ items where $price(A)$ and $tie(A)$ are small, yet $cc(A)=exp(m)$.

\begin{proposition}
There is a mechanism $A$ with $cc(A)=exp(m,n)$, $tie(A)=m$ and $price(A)=1$.
\end{proposition}
\begin{proof}
Let $f$ be the following two-player social choice function: player $1$ with valuation $v_1$ is never allocated any items. Player $2$ with valuation $v_2$ gets his maximum value bundle among all bundles $S$ of size $\frac m 2$ with $v_1(S)\geq 1$ and $v_2(S)\geq 1$. If there are several such bundles $S$, player $2$ gets the lexicographically first one. 

Notice that this social choice function can be implemented truthfully: the price of bundle $S$ is $1$ if $v_1(S)\geq 1$, and it is $\infty$ otherwise. Given the menu, tie breaking is also easy: player $2$ just announces his lexicographically first bundle among his set of profit-maximizing bundles. Also, player $1$ can announce the price of $S$ simply by sending the bit ``1'' if $v_1(S)\geq 1$ and ``0'' otherwise.

We now show that the communication complexity of every truthful mechanism that implements $f$ is $exp(m)$. We reduce again from the set disjointness problem. Let $t={m \choose {\frac m 2}}$. Let the valuation of Alice (player 1) be identically $0$ for bundles $S$ with $|S|<\frac m 2$ and identically $1$ for bundles $S$ with $|S|>\frac m 2$. For bundles $S$ with $|S|=\frac m 2$ the value of $S$ is determined by the value of $A_S$, where we assume some one-to-one and onto correspondence between the indices of $A$ and bundles of size $\frac m 2$. Bob's valuation (player 2) is defined similarly but with value $2$ instead of $1$ for non-zero bundles.

Notice that there is a bit $k$ with $A_k=B_k=1$ if and only if Bob has a bundle $S$ with a positive profit ($=1$). The communication complexity of $f$ is therefore ${m \choose {\frac m 2}}=exp(m)$.
\end{proof}

\subsubsection{Dropping $price(A)$}

We give an example for a truthful mechanism $A$ for three players and $m$ items where $tax(A)$ and $tie(A)$ are small, yet $cc(A)=exp(m)$. Notice the necessity of using more than two players, otherwise the price $S$ in the menu that is presented to one player can be easily computed by the other player, as the menu depends only on the valuation of the player that presents the menu.

\begin{proposition}
There is a mechanism $A$ with $cc(A)=exp(m,n)$, $tie(A)=0$ and $tax(A)=2$.
\end{proposition}
\begin{proof}
Let $f$ be the following three-player social choice function: players $1,2$ with valuation $v_1,v_2$ never get any items. Player $3$ with valuation $v_2$ can only get the bundle that contains item $a$ alone or the empty bundle. The price of item $a$ is determined according to the following rule, if there exists some bundle $S$, $|S|=\frac m 2$, with $v_1(S)=v_2(S)=1$, then the price of item $a$ is $1$. Else, the price of item $a$ is $2$.

Notice that $f$ can obviously be implemented by a truthful mechanism. We also have that $tax(A)=2$. We now show that the communication complexity of every truthful mechanism that implements $f$ is $exp(m)$. We again reduce from the set disjointness problem. Let $t={m \choose {\frac m 2}}$. Let the valuation of Alice (player 1) be identically $0$ for bundles $S$ with $|S|<\frac m 2$ and identically $1$ for bundles $S$ with $|S|>\frac m 2$. For bundles $S$ with $|S|=\frac m 2$ the value of $S$ is determined by the value of $A_S$, where we assume some one-to-one and onto correspondence between the indices of $A$ and bundles of size $\frac m 2$. Bob's valuation (player 2) is defined similarly.

Notice that there is a bit $k$ with $A_k=B_k=1$ if and only if the price of item $a$ is $1$. The communication complexity of $f$ is therefore ${m \choose {\frac m 2}}=exp(m)$.
\end{proof}

\end{document}